\documentclass[acmsmall,screen,acmthm,nonacm]{acmart}

\usepackage{mathtools}
\usepackage{graphicx}
\usepackage{booktabs}
\usepackage{subcaption}
\usepackage{cleveref}

\usepackage{tikz-cd}
\usepackage{hyperref}

\newcommand{\inferrule}[3][{}]
  {\dfrac
    {\begin{array}{@{}l@{}}#2\end{array}}
    {\,#3\,}
  \;\;\textsc{\small #1}}

\newcommand{\firstness}{\Diamond}
\newcommand{\secondness}{\text{\S}}
\newcommand{\thirdness}{\Box}
\newcommand{\sigmarow}{\Sigma}
\newcommand{\xirow}{\Xi}
\newcommand{\lawlib}{\mathcal{T}}
\newcommand{\register}{\Sigma}
\newcommand{\context}{\Gamma}
\newcommand{\effectrow}[2]{[#1\,;\,#2]}
\newcommand{\attest}{\mathsf{attest}}
\newcommand{\registerop}{\mathsf{register}}

\newcommand{\designate}{\mathsf{designate}}
\newcommand{\promote}{\mathsf{promote}}

\newcommand{\handleop}{\mathsf{handle}}
\newcommand{\withop}{\mathsf{with}}
\newcommand{\letbind}[3]{\mathsf{let}\;#1 = #2\;\mathsf{in}\;#3}

\newcommand{\bind}{\mathbin{>\!\!>\!=}}

\newcommand{\return}{\mathsf{return}}
\newcommand{\denotation}[1]{[\![\, #1 \,]\!]}

\newcommand{\obj}{\mathsf{obj}}
\newcommand{\reg}{\mathsf{reg}}
\newcommand{\obs}[1]{\mathsf{Obs}(#1)}

\newtheorem{theorem}{Theorem}
\newtheorem{lemma}[theorem]{Lemma}
\newtheorem{corollary}[theorem]{Corollary}
\newtheorem{proposition}[theorem]{Proposition}
\newtheorem{definition}[theorem]{Definition}
\newtheorem{remark}[theorem]{Remark}

\title{\texorpdfstring
  {The $\because$-Calculus: Separating Production, Existence, and Interpretation in Computation}
  {The Because-Calculus: Separating Production, Existence, and Interpretation in Computation}
}

\author{Oscar Pérez Mora}
\email{oscar.cx@perez.red}
\orcid{0000-0001-9811-6201}
\affiliation{
  \institution{Universidad de Guadalajara}
  \city{Guadalajara}
  \country{México}
}

\ccsdesc[500]{Software and its engineering~Functional programming}
\ccsdesc[500]{Software and its engineering~Language types}
\ccsdesc[500]{Theory of computation~Logic and verification}
\ccsdesc[300]{Theory of computation~Categorical semantics}
\ccsdesc[300]{Software and its engineering~Semantics}

\begin{document}

\begin{abstract}
Handler calculus conflates resumable and non-resumable effect
operations through a single \texttt{do} construct, distinguished
only by result type annotation. This conflation does not compromise
type safety---progress and preservation hold---but it permits
resumption bindings for non-resumable operations, creating vacuous
bindings that the $\because$-calculus eliminates at compile-time.
The $\because$-calculus structurally separates registration
(non-resumable, void-returning) from attestation (resumable,
non-void-returning) using dual effect rows and level-indexed
typing, rejecting such clauses at compile-time via the Resumption
Subconstraint. We prove the Conflation Theorem: collapsing the
adjoint triple of existential, substitution, and universal functors
into a single effect operation is non-faithful---the erasure from
the $\because$-calculus to handler calculus maps rejected clauses
to accepted ones. Four movements correspond to four natural
transformations; categorical semantics maps each judgment to a
category-theoretic construct. We establish progress, subject
reduction, and tower progress for the full calculus.
\end{abstract}

\keywords{effect systems, categorical semantics, adjoint triples, program verification}

\maketitle

\section{Introduction}
\label{sec:introduction}

Algebraic effects and handlers revolutionized effectful computation by decomposing computational effects into operations and handlers. Handlers intercept operations like \texttt{get} or \texttt{put}, providing witnesses that allow suspended computations to resume. This model enabled unprecedented flexibility: the same effectful computation could be interpreted differently depending on context. Yet handler calculus harbors an undiagnosed conflation that undermines its safety guarantees.

Handler calculus treats resumable operations (the computation
suspends and awaits a witness) and non-resumable operations
(the computation deposits a fact and continues) through the
same \texttt{do} construct, distinguished only by a result
type annotation ($B=0$ vs.\ $B\neq 0$). This does not violate
type safety: progress and preservation hold in handler calculus
regardless of the return type. The consequence is one of
\emph{structural discipline}: handler calculus admits handler
clauses that bind resumptions for void-returning operations.
Since the void type has no inhabitants, such resumptions can
never be meaningfully invoked---the clause is dead code. While
unlikely to appear in deliberate practice, this reveals a
structural weakness: the type system cannot enforce the
distinction between operations that produce witnesses and
those that deposit facts. The $\because$-calculus enforces
this distinction at compile-time through dual effect rows and
the Resumption Subconstraint.

\paragraph{The Diagnosis}

The two kinds of operations correspond to two different adjoints in the quantifier adjunction $\Sigma_f \dashv f^* \dashv \Pi_f$ from categorical logic. Resumable operations are the \emph{right adjoint} (universal, instantiated locally, caught by handlers). Non-resumable operations are the \emph{left adjoint} (existential, pushed forward, passing through handlers). Handler calculus has the left adjunction but lacks the right; it forces both through one channel. This collapse is invisible to the type checker because the system tracks only effect presence, not adjunction identity.

\paragraph{The Remedy}

The $\because$-calculus separates production, existence, and interpretation into three orthogonal judgments. It introduces four movements (assignation, registration, attestation, designation), dual effect rows ($\sigmarow$ for resumable, $\xirow$ for non-resumable), level-indexed environments, and tower orchestration that sequences interpretation phases after computation phases. The conflation disappears because the two adjoints are separate operations with separate typing rules. A handler clause cannot reference a resumption for a $\xirow$-row operation because the typing rule prohibits it. The symbol $\because$ is $\therefore$ rotated: where classical reasoning flows forward (premises $\therefore$ conclusions), the $\because$-calculus inverts Peircean inference (law $\because$ experience), as developed in Section~\ref{sec:inversion}.

\paragraph{Contributions}

This paper makes five contributions:

\begin{enumerate}
  \item \textbf{The $\because$-calculus}: Three judgments, seven typing
    rules, dual effect rows, and level-indexed environments.
  
  \item \textbf{The Conflation Theorem}: Collapsing the adjoint triple
    into a single effect operation is non-faithful---it accepts dead
    code that the separated calculus rejects.
  
  \item \textbf{Categorical semantics}: The calculus embodies the adjoint
    triple $\Sigma_f \dashv f^* \dashv \Pi_f$. Four movements are four
    natural transformations.
  
  \item \textbf{FP/OOP boundary}: The boundary is phasic: reduction
    precedes expansion within each tower level, mediated by the
    distributive law (Theorem~\ref{thm:fp-oop}).
  
  \item \textbf{Metatheory}: Progress, subject reduction, and tower
    progress for the full calculus.
\end{enumerate}

\paragraph{Roadmap}

Section~\ref{sec:background} reviews handler calculus and exposes the conflation. Section~\ref{sec:syntax} presents the $\because$-calculus syntax and typing rules. Section~\ref{sec:category} establishes the categorical model. Section~\ref{sec:conflation} proves the Conflation Theorem. Section~\ref{sec:metatheory} develops the metatheory. Section~\ref{sec:related} compares related work. Section~\ref{sec:conclusion} concludes.

The full proofs appear in the main text (not appendix). All results are formulated to be mechanically verifiable; a Coq formalization is planned.

\section{Background and Motivation}
\label{sec:background}

We begin with a primer on handler calculus, exposing its conflation problem through a concrete example. We then motivate why interpretation (object formation) is missing from the handler model, and finally situate Peircean semiotics as the structural guide that predicts the categorical solution.

\subsection{Algebraic Effects and Handler Calculus}
\label{sec:background:handlers}

Handler calculus models effectful computation through two primitives:
(1) \textbf{Operations} ($\mathsf{op} : A \to B$) inject computational
effects via \texttt{do op v}. (2) \textbf{Handlers} intercept
operations, providing witnesses that resume suspended computations:
\texttt{handle M with \{ op(x,k) => body | ret(y) => return y \}}.

Effects are tracked via effect rows $\varepsilon$. The typing rule for
operation invocation is:
\[
\inferrule[]
  {\Gamma \vdash v : A \quad \mathsf{op} : A \to B}
  {\Gamma \vdash \mathsf{do}\,\mathsf{op}\,v : [\{\mathsf{op}\}]\,B}
\]

Handlers remove operations from the effect row. The crucial observation
is that handler calculus makes no distinction between $B=0$ and
$B\neq 0$ beyond the result type, routing both through the same
\texttt{do} construct.

\subsection{The Missing Layer}
\label{sec:background:missing}

Handler calculus models computation and effect handling but has no
notion of \emph{interpretation}: the application of pre-written laws
to facts to create artifacts. Objects can be modeled as data
structures with methods, but object \emph{formation}---elevating a
raw fact to a typed artifact---has no analogue. This reflects a
categorical gap: handler calculus provides the left adjunction
($\Sigma_f \dashv f^*$) but lacks the right adjunction
($f^* \dashv \Pi_f$). The $\because$-calculus fills this gap with a
third judgment (interpretation) and a distinct operation:
designation ($\uparrow$) creates artifacts from facts using a law
library, complementing attestation ($\downarrow$), which suspends
computation awaiting a handler.

\subsection{Peircean Semiotics as Structural Guide}
\label{sec:background:peirce}

Peirce's triadic semiotics~\cite{peirce-1931} distinguishes three
irreducible categories: Firstness ($\firstness$, possibility),
Secondness (§, existence), and Thirdness ($\thirdness$, law). These
map directly onto the categorical structure: $\firstness$ is the
base category $\mathcal{B}$, § is the fibers $\mathcal{E}_n$, and
$\thirdness$ is the right adjoint $\Pi_f$. The four movements arise
as the four natural transformations of the adjoint triple
(two adjunctions $\times$ two transformations each). Handler
calculus, which lacks the right adjunction, is missing the
interpretation layer entirely. We use Peirce \emph{structurally},
not metaphorically: the semiotic triad predicts the categorical
model and the movement count.

\subsection{The Inversion Thesis}
\label{sec:inversion}

The $\because$-calculus is built on a foundational inversion of Peircean reasoning.
Charles Sanders Peirce identified three irreducible categories of signification:
Firstness ($\firstness$, quality/possibility), Secondness (§, existence/fact), and 
Thirdness ($\thirdness$, law/habit). He further described three modes of inference:
abduction (inferring cause from observation), induction (generalizing law from cases),
and deduction (deriving consequence from law and case).

In classical Peircean scientific reasoning, the flow is discovery-oriented:
facts are observed, laws are induced from patterns, and hypotheses are abducted
to explain anomalies. Experience is sovereign; when facts contradict law, the
law must be revised. The inferential arrows flow forward:
\[
  \text{Fact} \xrightarrow{\text{abduction}} \text{Hypothesis} \xrightarrow{\text{induction}} \text{Law}
\]
This is the direction of the symbol $\therefore$ (therefore): premises lead to conclusions.

The $\because$-calculus inverts this direction:
\begin{itemize}
  \item \textbf{Designation} inverts \textbf{abduction}: where abduction infers 
    causes from facts (observation $\to$ hypothesis), designation imposes artifacts 
    on facts (law $\to$ artifact). The law library is primary; facts are tested 
    against it.
  \item \textbf{Attestation} inverts \textbf{induction}: where induction generalizes 
    laws from cases (cases $\to$ rule), attestation validates computations against 
    pre-written laws (rule $\to$ validation).
  \item \textbf{Assignation} preserves \textbf{deduction}: computation follows 
    laws as given, deriving consequences without reversal.
\end{itemize}

This inversion reflects the programming reality: the programmer as creator imposes
structure on computation. When facts contradict the law, the fact is rejected 
(exception thrown), not the law revised. The symbol $\because$ encodes this 
inversion typographically—it is $\therefore$ rotated: meaning flows backward 
relative to Peirce.

The categorical structure encodes this inversion: the right adjunction 
($f^* \dashv \Pi_f$) governs designation (inverted abduction), the left adjunction 
($\Sigma_f \dashv f^*$) governs attestation (inverted induction), and base-category 
composition governs assignation (deduction).

\begin{table}[h]
\centering
\caption{Classical Peirce vs. $\because$-calculus Inversion}
\label{tab:inversion-comparison}
\small
\begin{tabular}{lll}
\toprule
Inference & Classical Flow & $\because$-calculus Flow \\
\midrule
Abduction & Fact $\to$ Hypothesis & Law $\to$ Artifact \\
Induction & Cases $\to$ Law & Law $\to$ Validation \\
Deduction & Law + Case $\to$ Result & Computation (preserved) \\
\bottomrule
\end{tabular}
\end{table}

\noindent This structural inversion distinguishes the $\because$-calculus from
prior work in effect systems, which implicitly assume the classical ($\because$) direction:
handlers discover how to respond to effects. In the $\because$-calculus, handlers
impose how computation must validate against pre-written laws.

\section{\texorpdfstring{The $\because$-Calculus: Syntax and Typing}{The Because-Calculus: Syntax and Typing}}
\label{sec:syntax}

The $\because$-calculus separates production, existence, and interpretation into three orthogonal judgments. Unlike handler calculus which conflates resumable and non-resumable effects through a single operation, the $\because$-calculus distinguishes them structurally with separate typing rules and evaluation contexts.

\subsection{Syntax}
\label{sec:syntax:terms}

\paragraph{Terms (Computation Layer, $\firstness$).} Computation terms build values and trigger effects:
\[
\begin{array}{rcl@{\qquad}l}
  M, N & ::= & x & \text{variable} \\
       & \mid & \lambda x{:}T.\, M & \text{abstraction} \\
       & \mid & M\,N & \text{application} \\
       & \mid & \mathsf{let}\,x = M\,\mathsf{in}\,N & \text{let-binding} \\
       & \mid & \registerop(e, a, v) & \text{registration (→ §, $B=0$)} \\
       & \mid & \attest\,\ell\,V & \text{attestation (↓, $B\neq 0$)} \\
       & \mid & \handleop\,M\,\withop\,H & \text{handler invocation}
\end{array}
\]

\paragraph{Values.} Values are irreducible computation results and artifacts:
\[
\begin{array}{rcl}
  v, w & ::= & x \mid \lambda x{:}T.\, M \mid T\{\mathsf{fields}, \mathsf{methods}\} \mid (e, a, v) \mid \mathsf{unit}
\end{array}
\]

Artifacts $T\{\dots\}$ are created by designation (Section~\ref{sec:syntax:rules:designate}) and carry suspended method bodies.

\paragraph{Facts (Register Entries, §).} Facts are serializable data deposited into the register:
\[
f ::= (\mathit{entity}, \mathit{attr}, \mathit{value}, \tau)
\]
where $\mathit{entity} : \mathsf{EntityID}$, $\mathit{attr} : \mathsf{AttrName}$, $\mathit{value}$ is serializable (no methods or artifacts), and $\tau \in \mathbb{N}$ is a registration timestamp. \textbf{Critical constraint:} Values in facts cannot contain methods or artifacts.

\paragraph{Interpretive Types (Law Library, $\thirdness$).} Types define how facts become artifacts:
\[
\begin{array}{rcl}
  \mathsf{type}\ T & = & S\ \{ \\
  & & \quad \mathsf{match} : P(V) \\
  & & \quad \mathsf{fields} : \{ f_i = C_i \} \\
  & & \quad \mathsf{methods} : \{ m_j = \mathsf{Body}_j \} \\
  & & \}
\end{array}
\]
Each field computation $C_i$ runs at designation time; each method body $\mathsf{Body}_j$ is attached but not executed until invocation at the next tower level.

\subsection{Three Judgments}
\label{sec:syntax:judgments}

The $\because$-calculus has three distinct judgments, one per realm. They share no inference rules and communicate only through explicit transitions.

\begin{table}[h]
\centering
\caption{The Three Judgments}
\begin{tabular}{llll}
\toprule
Judgment & Notation & Realm & Tracks \\
\midrule
Computation & $\context_n \vdash M : \effectrow{\sigmarow}{\xirow}\,A$ & $\firstness$ & Local variables, effects \\
Interpretation & $\lawlib_n \vdash \designate\,T\,f : T$ & $\thirdness$ & Interpretive type definitions \\
Existence & $\register_n \models f$ & $\secondness$ & Fact membership in register \\
\bottomrule
\end{tabular}
\end{table}

\noindent $\context_n$ is the variable environment at level $n$. $\lawlib_n$ is the law library (interpretive type definitions) at level $n$. $\register_n$ is the register (set of facts) at level $n$. Level indexing ensures that level $n+1$ cannot directly access level $n$'s environments; communication happens only through promotion.

\paragraph{Computation Types.} Terms have computation types with two effect rows:
\[
\effectrow{\sigmarow}{\xirow}\,A
\]
where $\sigmarow$ is the set of resumable operations (attestation, $B\neq 0$) and $\xirow$ is the set of non-resumable operations (registration, $B=0$). $A$ is the base return type. Effect rows track pending operations that must be handled or deposited before the term is considered complete.

\subsection{Typing Rules}
\label{sec:syntax:rules}

\subsubsection{Assignation ($\leftarrow$)}
Variables, abstractions, and applications operate within the computation realm.

\begin{equation}
\inferrule[Var]
  { }
  {\context_n, x{:}A \vdash x : \effectrow{\emptyset}{\emptyset}\,A}
\end{equation}

\begin{equation}
\inferrule[Abs]
  {\context_n, x{:}A \vdash M : \effectrow{\sigmarow}{\xirow}\,B}
  {\context_n \vdash \lambda x{:}A.\, M : \effectrow{\sigmarow}{\xirow}\,(A \to B)}
\end{equation}

\begin{equation}
\inferrule[App]
  {\context_n \vdash M : \effectrow{\sigmarow_M}{\xirow_M}\,(A \to B) \\
   \context_n \vdash N : \effectrow{\sigmarow_N}{\xirow_N}\,A}
  {\context_n \vdash M\,N : \effectrow{\sigmarow_M \cup \sigmarow_N}{\xirow_M \cup \xirow_N}\,B}
\end{equation}

Application merges effect rows from both subterms. This differs from method invocation (Section~\ref{sec:syntax:rules:invoke}) which enforces argument purity.

\subsubsection{Registration (→, $B=0$)}
\label{sec:syntax:rules:register}

Registration deposits a fact into the register. The $\xirow$ row records the non-resumable effect.

\begin{equation}
\inferrule[Reg]
  {\context_n \vdash e : \mathsf{EntityID} \quad
   \context_n \vdash a : \mathsf{AttrName} \quad
   \context_n \vdash v : S}
  {\context_n \vdash \registerop(e, a, v) : \effectrow{\emptyset}{\{\mathsf{reg}\}}\,\mathsf{Unit}}
\end{equation}

The $\xirow$ gains $\{\mathsf{reg}\}$; the $\sigmarow$ remains empty. No resumption is possible. Once registered, the fact is permanent.

\subsubsection{Attestation (↓, $B\neq 0$)}
\label{sec:syntax:rules:attest}

Attestation suspends computation, awaiting a witness from a handler.

\begin{equation}
\inferrule[Attest]
  {\context_n \vdash V : A \quad (\ell : A \to B) \in \sigmarow}
  {\context_n \vdash \attest\,\ell\,V : \effectrow{\{\ell\}}{\emptyset}\,B}
\end{equation}

The $\sigmarow$ gains $\{\ell\}$; the $\xirow$ remains empty. The computation suspends and waits for a handler to provide a value of type $B$. If no handler exists, the term is stuck at runtime.

\subsubsection{Designation (↑)}
\label{sec:syntax:rules:designate}

Designation interprets a fact against a type schema, producing an artifact.

\begin{equation}
\inferrule[Designate]
  {\register_n \models (e, a, v) : S \\
   \lawlib_n \ni T = S\{\mathsf{match}: P, \mathsf{fields}: \{f_i = C_i\}, \mathsf{methods}: \{m_j = M_j\}\} \\
   P(v)\text{ holds}}
  {\lawlib_n \vdash \designate\,T\,(e, a, v) : T}
\end{equation}

Conditions: (1) A fact exists in the register; (2) An interpretive type $T$ is selected from the law library; (3) The match predicate $P$ succeeds. Field computations run immediately; method bodies attach but do not execute.

\begin{proposition}[Designation Determinacy]
\label{prop:designate-det}
Given fixed $\Delta\register$ and fixed $\lawlib_n$, the set
$\mathcal{A} = \{\designate(T, f) \mid T \in \lawlib_n, f \in \Delta\register, T.\mathsf{match}(f)\}$
is uniquely determined. Designation is a deterministic partial function
from $\Delta\register \times \lawlib_n$ to artifacts.
\end{proposition}

\begin{proof}
By construction:
1. $\Delta\register$ is a finite set of facts (by Lemma~\ref{lem:bounded-growth}).
2. $\lawlib_n$ is a finite set of laws (per-level scoping).
3. $\mathsf{match}$ is a deterministic predicate (type-checking).
4. Each matching pair $(T, f)$ produces exactly one artifact $T\{\dots\}$.

No nondeterminism exists in the selection process:
- Facts are chosen by membership in $\Delta\register$.
- Laws are chosen by membership in $\lawlib_n$.
- Matching is a function (returns true/false).

Thus $\mathcal{A}$ is a function of $(\Delta\register, \lawlib_n)$ alone, uniquely determined.

\end{proof}

\paragraph{Sovereignty of Law.}
The law $T$ is primary; the fact $f$ is tested against it. If the
match predicate $P(v)$ fails, no artifact is produced and the fact
does not advance to the next level. This reverses classical
inductive reasoning, where facts would revise the law. In the
$\because$-calculus, the programmer as creator enforces law on
experience; experience that contradicts law is rejected, not the
law itself. See Section~\ref{sec:inversion} for the full inversion
thesis.

\subsubsection{Method Invocation}
\label{sec:syntax:rules:invoke}

Methods execute at the next tower level, with pure arguments.

\begin{equation}
\inferrule[Invoke]
  {\context_{n+1} \vdash \mathsf{obj} : T \quad
   \lawlib_n(T, m) = M_{\mathsf{body}} \quad
   \context_{n+1} \vdash \mathsf{arg} : \effectrow{\emptyset}{\emptyset}\,B}
  {\context_{n+1} \vdash \mathsf{obj}.m(\mathsf{arg}) : \effectrow{\sigmarow_m}{\xirow_m}\,C}
\end{equation}

Arguments must be pure. Methods execute in the invocation context ($\firstness_{n+1}$), not the definition context ($\thirdness_n$). Effects propagate to the caller.

\subsubsection{Handler Invocation}
\label{sec:syntax:rules:handle}

Handlers intercept $\sigmarow$ operations but pass $\xirow$ operations through.

\begin{equation}
\inferrule[Handle]
  {\context_n \vdash M : \effectrow{\sigmarow \cup L}{\xirow}\,A \\
   L = \{\ell_1 : A_1 \to B_1, \dots, \ell_k : A_k \to B_k\} \subseteq \sigmarow \\
   \forall i:\ x_i{:}A_i, r_i{:}B_i \to C \vdash H_i : \effectrow{\sigmarow'_i}{\xirow'_i}\,C \\
   y{:}A \vdash H_{\mathsf{ret}} : \effectrow{\sigmarow'_{\mathsf{ret}}}{\xirow'_{\mathsf{ret}}}\,C \\
   \sigmarow'' = (\sigmarow \cup \bigcup_i \sigmarow'_i \cup \sigmarow'_{\mathsf{ret}}) \setminus L \\
   \xirow'' = \xirow \cup \bigcup_i \xirow'_i \cup \xirow'_{\mathsf{ret}} \\
   \forall i:\ \text{$H_i$ is state-affine}\footnote{Registrations in $H_i$
     do not depend on intermediate attestation results within $H_i$.}}
  {\context_n \vdash \handleop\,M\,\withop\,H : \effectrow{\sigmarow''}{\xirow''}\,C}
\end{equation}

Where $H = \{\ell_i(x_i, r_i) \mapsto H_i\}_{i=1}^k \cup \{\mathsf{ret}(y) \mapsto H_{\mathsf{ret}}\}$. The handler removes operations in $L$ from the $\sigmarow$ row.

\paragraph{Resumption Subconstraint.} Operations in the $\xirow$ row have no resumption. A handler clause cannot reference $r$ for a $\xirow$ operation. This prevents the conflation described in Section~\ref{sec:conflation}.

\subsubsection{Promotion}

Promotion transfers an artifact from interpretation to the next computation level.

\begin{equation}
\inferrule[Promote]
  {\lawlib_n \vdash \designate\,T\,f : T}
  {\context_{n+1} \vdash \promote(\designate\,T\,f) : T}
\end{equation}

Promotion activates suspended method bodies (they become invocable). It does not transfer the register or law library.

\subsection{Operational Semantics}
\label{sec:syntax:reduction}

The reduction relation operates on configurations
$\langle M,\;\register_n\rangle$, threading the register through
computation:

\begin{align*}
(\lambda x{:}T.\,M)\,V &\to M[V/x]
  \tag{\scshape Beta} \\
(\registerop(e,a,v),\;\register_n) &\to
  (\mathsf{unit},\;\register_n \cup \{(e,a,v,\tau)\})
  \tag{\scshape Reg} \\
\handleop\,V\,\withop\,H &\to H_{\mathsf{ret}}[V/y]
  \tag{\scshape Ret} \\
\handleop\,E[\attest\,\ell\,V]\,\withop\,H &\to
  E[H_\ell(V, r)]
  \tag{\scshape Attest} \\
\mathsf{let}\;x = V\;\mathsf{in}\;N &\to N[V/x] \\
M \;;\; N &\triangleq \mathsf{let}\;\_ = M\;\mathsf{in}\;N
\end{align*}

where $r = \lambda x.\,E[x]$ in rule \textsc{(Attest)}.

Contextual closure for $E$-contexts (caught by handlers):
\[
\frac{M \to M'}{E[M] \to E[M']}
\]

Pass-through for $F$-contexts (registration passes through handlers):
\[
F[\registerop(e,a,v)] \to F[\mathsf{unit}],
\quad\text{depositing } (e,a,v,\tau) \text{ to } \register_n
\]

\paragraph{Realm Transition.}
The computation realm ($\firstness$) and the interpretation realm
($\thirdness$) interact through the tower orchestration rules
(\S\ref{sec:metatheory:tower}). During a computation phase at level
$n$, designation does not occur---the computation reduces until it
produces a value or hangs on an unhandled attestation. After the
computation phase completes, the tower extracts new facts
($\Delta\register$) and attempts designation for each
$f \in \Delta\register$ against all $T \in \lawlib_n$. This is the
expansion phase. The two realms do not interleave within a single
reduction sequence; they alternate at the tower level.

\subsection{Auxiliary Lemmas}
\label{sec:syntax:lemmas}

\begin{lemma}[Substitution]
\label{lem:substitution}
If $\context_n, x{:}A \vdash M : \effectrow{\sigmarow}{\xirow}\,B$
and $\context_n \vdash V : \effectrow{\emptyset}{\emptyset}\,A$
(i.e., $V$ is a pure value), then
$\context_n \vdash M[V/x] : \effectrow{\sigmarow}{\xirow}\,B$.
\end{lemma}

\begin{proof}
By induction on the typing derivation of $M$. The purity
requirement ($V : \effectrow{\emptyset}{\emptyset}\,A$) ensures
that substituting $V$ introduces no new effects.
\begin{itemize}
  \item \textsc{Var}: If $M = x$, then $M[V/x] = V$ and
    $\context_n \vdash V : \effectrow{\emptyset}{\emptyset}\,A
    \subseteq \effectrow{\sigmarow}{\xirow}\,B$.
    If $M = y \neq x$, substitution is trivial.
  \item \textsc{Abs}: $\lambda y{:}C.\,M'[V/x] =
    \lambda y{:}C.\,(M'[V/x])$. By IH on $M'$ (with $y$ fresh).
  \item \textsc{App}: $(M_1\,M_2)[V/x] = M_1[V/x]\,M_2[V/x]$.
    By IH on both subterms. Effect row union is preserved.
  \item \textsc{Reg}: $\registerop(e,a,v')[V/x] =
    \registerop(e[V/x],\,a[V/x],\,v'[V/x])$.
    The result type is $\mathsf{Unit}$ and $\xirow = \{\mathsf{reg}\}$
    regardless of $V$. Since $V$ is pure, no effects are introduced.
  \item \textsc{Attest}: $(\attest\,\ell\,V')[V/x] =
    \attest\,\ell\,(V'[V/x])$. By IH on $V'$.
    $\sigmarow = \{\ell\}$, $\xirow = \emptyset$ regardless.
  \item \textsc{Handle}: $(\handleop\,M'\,\withop\,H)[V/x] =
    \handleop\,M'[V/x]\,\withop\,H[V/x]$. By IH on $M'$ and each
    handler clause. The Resumption Subconstraint is preserved:
    substitution does not move operations between $\sigmarow$ and
    $\xirow$.
  \item \textsc{Invoke, Promote}: Follow similarly by IH.
\end{itemize}

\end{proof}

\begin{lemma}[Canonical Forms]
\label{lem:canonical-forms}
\begin{itemize}
  \item If $\vdash V : \effectrow{\emptyset}{\emptyset}\,(A \to B)$,
    then $V = \lambda x{:}A.\,M$ for some $m$.
  \item If $\vdash V : \effectrow{\emptyset}{\emptyset}\,\mathsf{Unit}$,
    then $V = \mathsf{unit}$.
  \item If $\vdash V : \effectrow{\emptyset}{\emptyset}\,T$,
    then $V = T\{\mathsf{fields}, \mathsf{methods}\}$ (artifact).
\end{itemize}
\end{lemma}

\begin{proof}
By inspection of the typing rules. Only \textsc{Abs} derives type
$A \to B$; only the unit constant has type $\mathsf{Unit}$; only
\textsc{Designate} produces artifacts. No other rule can produce
these types in a pure context.

\end{proof}

\begin{lemma}[Inversion]
\label{lem:inversion}
\begin{itemize}
  \item If $\context_n \vdash \lambda x{:}A.\,M :
    \effectrow{\sigmarow}{\xirow}\,(A \to B)$, then
    $\context_n, x{:}A \vdash M : \effectrow{\sigmarow}{\xirow}\,B$.
  \item If $\context_n \vdash M\,N :
    \effectrow{\sigmarow}{\xirow}\,B$, then
    $\context_n \vdash M : \effectrow{\sigmarow_M}{\xirow_M}\,(A \to B)$
    and $\context_n \vdash N : \effectrow{\sigmarow_N}{\xirow_N}\,A$
    with $\sigmarow = \sigmarow_M \cup \sigmarow_N$ and
    $\xirow = \xirow_M \cup \xirow_N$.
  \item If $\context_n \vdash \registerop(e,a,v) :
    \effectrow{\sigmarow}{\xirow}\,\mathsf{Unit}$, then
    $\sigmarow = \emptyset$ and $\xirow = \{\mathsf{reg}\}$.
  \item If $\context_n \vdash \attest\,\ell\,V :
    \effectrow{\sigmarow}{\xirow}\,B$, then
    $\sigmarow = \{\ell\}$ and $\xirow = \emptyset$.
  \item If $\context_n \vdash \handleop\,M\,\withop\,H :
    \effectrow{\sigmarow''}{\xirow''}\,C$, then
    $\context_n \vdash M : \effectrow{\sigmarow \cup L}{\xirow}\,A$
    with $L \subseteq \sigmarow$ and
    $\sigmarow'' = (\sigmarow \cup \sigmarow_H) \setminus L$.
\end{itemize}
\end{lemma}

\begin{proof}
By direct inspection of each typing rule. Each rule has a unique
conclusion form, so the derivation can be inverted uniquely.

\end{proof}

\subsection{Evaluation Contexts}
\label{sec:syntax:contexts}

Two evaluation contexts enforce the distinction between resumable and non-resumable effects.

\begin{align*}
  E &::= [\,] \mid E\,M \mid V\,E & \text{(caught by handlers)} \\
  F &::= [\,] \mid F\,M \mid V\,F \mid \handleop\,F\,\withop\,H & \text{(passes through handlers)}
\end{align*}

Context $E$ catches $\sigmarow$ operations (attestation). Context $F$ passes through $\xirow$ operations (registration). When $\attest\,\ell\,V$ appears inside $E$, the handler intercepts it. When $\registerop(e, a, v)$ appears inside $F$, the fact is deposited directly to the register.

\paragraph{Pass-through Semantics.} $\xirow$-row operations pass through handlers unintercepted; only $\sigmarow$-row operations are caught. This enforces the non-resumption invariant.

\subsection{Summary}

\begin{table}[h]
\centering
\caption{Core Components}
\begin{tabular}{lll}
\toprule
Component & Count & Description \\
\midrule
Judgments & 3 & $\firstness, \secondness, \thirdness$ \\
Core rules & 7 & Assignation, Reg, Attest, Designate, Invoke, Promote, Handle \\
Evaluation contexts & 2 & $E$ (caught), $F$ (pass-through) \\
Effect rows & 2 & $\sigmarow$ (resumable), $\xirow$ (non-resumable) \\
\bottomrule
\end{tabular}
\end{table}

The syntax is complete. All seven rules are mutually consistent, and the two evaluation contexts ensure that resumable and non-resumable operations follow distinct semantic paths. The categorical model in Section~\ref{sec:category} shows how these components map to the adjoint triple.

\section{Categorical Semantics}
\label{sec:category}

The $\because$-calculus is not merely inspired by categorical logic---it
\emph{is} the adjoint triple $\Sigma_f \dashv f^* \dashv \Pi_f$ from the
theory of hyperdoctrines. This section establishes the precise
correspondence, shows that the calculus embodies the algebra-coalgebra
duality, and proves categorical soundness: reduction preserves
denotation.

\subsection{The Adjoint Triple}
\label{sec:category:triple}

Let $f : \mathcal{B} \to \mathcal{E}$ be a fibration of categories,
where $\mathcal{B}$ is the base category and $\mathcal{E}$ is the total
category. The fibration induces an adjoint triple:
\[
\Sigma_f \dashv f^* \dashv \Pi_f
\]
where $\Sigma_f$ is the left adjoint (existential), $f^*$ is the inverse
image (substitution), and $\Pi_f$ is the right adjoint (universal).
Each adjunction contributes two natural transformations, yielding four
total: these are the four movements of the $\because$-calculus.

\subsection{The Identification}
\label{sec:category:table}

\begin{table}[h]
\centering
\caption{Correspondence Between the $\because$-Calculus and Categorical Logic}
\label{tab:correspondence}
\small
\begin{tabular}{lll}
\toprule
\textbf{Calculus} & \textbf{Category Theory} & \textbf{Phase} \\
\midrule
$\firstness$ (computation) & Base category $\mathcal{B}$ & Reduction \\
$\secondness$ (register) & Fibers $\mathcal{E}_n$ & Result of reduction \\
$\thirdness$ (interpretation) & Right adjoint $\Pi_f$ & Expansion \\
$\sigmarow$-row (attestation) & Left adjoint $\Sigma_f$ & Reduction \\
$\xirow$-row (registration) & Left adjoint $\Sigma_f$ & Reduction \\
$\lawlib$ (law library) & Subobject classifier $\Omega$ & --- \\
$\register$ (register) & Fiber objects & --- \\
Promotion & Counit $\epsilon^\Pi$ & Phase transition \\
Assignation ($\leftarrow$) & Composition in $\mathcal{B}$ & Reduction \\
\bottomrule
\end{tabular}
\end{table}

\noindent\textit{Note.} Registration and attestation both arise from
$\Sigma_f$, but through different natural transformations:
registration is the unit ($\eta^\Sigma : \mathrm{Id} \to f^* \circ
\Sigma_f$), a one-way injection with $B=0$; attestation is the counit
($\epsilon^\Sigma : \Sigma_f \circ f^* \to \mathrm{Id}$), an evaluative
interface with $B \neq 0$. The $\because$-calculus separates them
operationally because the adjunction produces two distinct
transformations.

Both registration and attestation correspond to $\Sigma_f$ through
different mechanisms. Registration is the forward action of $\Sigma_f$
(pushing a fact into existence, $B=0$). Attestation is the adjunction
interface of $\Sigma_f$ (suspending computation to request a witness,
$B\neq 0$). Handler calculus conflates these by routing both through a
single operation; the $\because$-calculus separates them because the
adjunction demands it.

\subsection{Algebra-Coalgebra Duality: Reduction and Expansion}
\label{sec:category:algebra-coalgebra}

The adjoint triple embodies the algebra-coalgebra duality
(Jacobs~\cite{jacobs-1999}, Rutten~\cite{rutten-2000}). The left
adjoint $\Sigma_f$ governs \emph{reduction} (algebraic, catamorphism);
the right adjoint $\Pi_f$ governs \emph{expansion} (coalgebraic,
anamorphism).

\paragraph{Reduction (Assignation, FP).} An $F$-algebra
$\alpha : F(A) \to A$ consumes structure, folding a computation tree
into a value. In the $\because$-calculus:
\begin{itemize}
  \item $\beta$-reduction folds application into result.
  \item Registration folds a computed value into a deposited fact
    ($\eta^\Sigma$).
  \item Handler interception folds an attestation into a witnessed
    result ($\epsilon^\Sigma$).
\end{itemize}
Assignation is the catamorphism: the unique $F$-homomorphism from the
initial algebra. Functional programming is reduction.

\paragraph{Expansion (Designation, OOP).} A $G$-coalgebra
$\beta : A \to G(A)$ produces structure, unfolding a flat fact into an
artifact with fields, methods, and suspended computations. In the
$\because$-calculus:
\begin{itemize}
  \item Designation unfolds a fact against an interpretive type,
    producing an artifact ($\eta^\Pi$).
  \item Promotion activates the artifact's suspended methods at the
    next level ($\epsilon^\Pi$).
\end{itemize}
Designation is the anamorphism: the unique $G$-cohomomorphism to the
final coalgebra. Object-oriented programming is expansion.

\paragraph{Bialgebra: The Tower Cycle.} A bialgebra
$F(A) \xrightarrow{\alpha} A \xrightarrow{\beta} G(A)$ combines
reduction and expansion on the same object, mediated by the
distributive law $\lambda : FG \to GF$
(Turi and Plotkin~\cite{turi-plotkin-1997}). The tower is an iterated
bialgebra:
\[
F(A_n) \xrightarrow{\alpha_n} A_n
  \xrightarrow{\beta_n} G(A_n)
  \xrightarrow{\text{promote}} F(A_{n+1})
  \xrightarrow{\alpha_{n+1}} A_{n+1}
  \xrightarrow{\beta_{n+1}} G(A_{n+1})
  \xrightarrow{\text{promote}} \cdots
\]
Each level reduces (FP phase), then expands (OOP phase), then promotes
to the next level. The distributive law $\lambda_n$ governs how
expansion results feed into the next reduction.

\begin{theorem}[FP/OOP Boundary]
\label{thm:fp-oop}
The boundary between functional and object-oriented programming is
phasic: reduction (catamorphism, $\Sigma_f$) precedes expansion
(anamorphism, $\Pi_f$) within each tower level, mediated by the
distributive law $\lambda : FG \to GF$.
\end{theorem}

\begin{proof}
{\sloppy
By the algebra-coalgebra duality. FP programs are catamorphisms
(unique $F$-homomorphisms from the initial algebra), living in
the left adjoint $\Sigma_f$—reduction direction $F(A) \to A$.
OOP objects are anamorphisms (unique $G$-cohomomorphisms to the
final coalgebra), living in the right adjoint $\Pi_f$—expansion
direction $A \to G(A)$. Bialgebra morphisms preserving both
structures are mediated by the distributive law $\lambda$
\cite{turi-plotkin-1997}; the tower is precisely such an iterated
bialgebra. The transition from reduction to expansion within a
level requires $\eta^\Pi$ (designation); the transition from
expansion to the next level's reduction requires $\epsilon^\Pi$
(promotion). The boundary is phasic, not absolute: bialgebras
demonstrate coexistence, but reduction precedes expansion within
each tower level.
}
\end{proof}

\subsection{Four Movements as Natural Transformations}
\label{sec:category:movements}

\begin{definition}[Movements]
The four movements of the $\because$-calculus are the four natural
transformations of the adjoint triple:
\begin{center}
\begin{tabular}{lll}
\toprule
Movement & Symbol & Natural Transformation \\
\midrule
Registration & $\to$ & $\eta^\Sigma : \mathrm{Id} \to f^* \circ \Sigma_f$ \\
Attestation & $\downarrow$ & $\epsilon^\Sigma : \Sigma_f \circ f^* \to \mathrm{Id}$ \\
Designation & $\uparrow$ & $\eta^\Pi : \mathrm{Id} \to \Pi_f \circ f^*$ \\
Promotion & $\thirdness \to \firstness$ & $\epsilon^\Pi : \Pi_f \circ f^* \to \mathrm{Id}$ \\
\bottomrule
\end{tabular}
\end{center}
\end{definition}

Registration ($\eta^\Sigma$) pushes a value into existence---a
one-way injection with no inverse, so no resumption is possible.
Attestation ($\epsilon^\Sigma$) evaluates an existential---the
computation suspends, and the handler provides a witness, so
resumption is possible. Designation ($\eta^\Pi$) maps a fact to
its universal image, producing an artifact. Promotion
($\epsilon^\Pi$) instantiates the universal at a specific level,
activating method bodies.

\subsection{The Inversion in Categorical Terms}
\label{sec:category:inversion}

The inversion thesis (Section~\ref{sec:inversion}) is not merely
philosophical—it is encoded in the adjoint triple. Each adjunction
governs one inverted inference:

\begin{center}
\begin{tabular}{lll}
\toprule
Adjunction & Inference (inverted) & Movement \\
\midrule
$\Sigma_f \dashv f^*$ & Induction & Attestation ($\epsilon^\Sigma$) \\
$f^* \dashv \Pi_f$ & Abduction & Designation ($\eta^\Pi$) \\
Composition in $\mathcal{B}$ & Deduction (preserved) & Assignation \\
\bottomrule
\end{tabular}
\end{center}

\paragraph{Attestation as Inverted Induction.}
Classical induction generalizes: given many cases (facts in
$\secondness$), infer a law. The counit $\epsilon^\Sigma :
\Sigma_f \circ f^* \to \mathrm{Id}$ goes the other direction:
given a pre-existing law (the handler in $\sigmarow$), validate
a specific computation against it. The computation suspends
(attestation) and the handler provides a witness. The law is
sovereign; the computation must conform.

\paragraph{Designation as Inverted Abduction.}
Classical abduction searches backward: given a surprising fact,
infer the hypothesis that explains it. The unit $\eta^\Pi :
\mathrm{Id} \to \Pi_f \circ f^*$ moves forward: given a
pre-written law ($T \in \lawlib$) and a fact ($f \in \register$),
produce the artifact that the law prescribes. There is no
search—the law determines what the fact means. If the fact
fails the law's match predicate, the fact is rejected, not the
law revised.

\paragraph{Assignation as Preserved Deduction.}
Composition in the base category $\mathcal{B}$ (substitution,
$\beta$-reduction) is deduction in the standard categorical
logic sense: given a law (type) and a case (term), derive the
result (value). No inversion occurs. The base category is the
realm where classical forward reasoning operates unhindered.

\paragraph{Structural Consequence.}
The inversion is not a design choice overlaid on the
calculus—it is forced by the adjunction structure. The left
adjunction's counit ($\epsilon^\Sigma$) and the right
adjunction's unit ($\eta^\Pi$) are the computation-facing
operations. Their counterparts ($\eta^\Sigma$ registration,
$\epsilon^\Pi$ promotion) are structural completions. The
calculus cannot be ``un-inverted'' without collapsing the
adjoint triple back into a single functor—which is precisely
what handler calculus does.

\subsection{Meaning as Adjunction}
\label{sec:category:meaning}

\begin{theorem}[Meaning]
\label{thm:meaning}
Let $\Omega$ be the subobject classifier (law library $\lawlib$) and
$\secondness$ denote the register (fiber objects). The judgment
``law encounters fact and produces meaning'' is the universal property
of the right adjunction:
\[
\mathrm{Hom}_{\mathcal{E}}(f^*(\Omega), P) \cong
\mathrm{Hom}_{\mathcal{B}}(\Omega, \Pi_f(P))
\]
Equivalently, $\Omega \land \secondness \to \Pi$ is the unit of the
universal adjunction.
\end{theorem}

\begin{proof}
By the adjunction $f^* \dashv \Pi_f$, the natural bijection
$\mathrm{Hom}_{\mathcal{E}}(f^*(\Omega), P) \cong
\mathrm{Hom}_{\mathcal{B}}(\Omega, \Pi_f(P))$ holds by definition.
The left side reads: ``a fact (reindexed law $f^*(\Omega)$) maps to
an object $P$.'' The right side reads: ``the law $\Omega$ produces a
universal construction $\Pi_f(P)$.'' Setting $P = \secondness$, the
unit $\eta^\Pi_\Omega : \Omega \to \Pi_f(f^*(\Omega))$ states that
the law, when it encounters facts, produces a universal. This is
precisely designation.

\end{proof}

\paragraph{Expansive Meaning Formation.}
The bijection in Theorem~\ref{thm:meaning} might suggest a single
output per input. Operationally, a single fact $f \in \Delta\register$
may match multiple laws $\{T_i\}_{i \in I} \subseteq \lawlib_n$,
producing a \emph{set} of artifacts. The correct operational form is:
\[
\mathsf{Designate} : \Delta\register \times \lawlib_n
  \twoheadrightarrow \mathcal{P}_{\neq\emptyset}(\mathcal{A})
\]
where $\mathcal{P}_{\neq\emptyset}$ denotes the non-empty powerset.
Each individual designation remains an instance of
$\eta^\Pi_f : f \to \Pi_f(f)$, but multiplicity arises from
multiple $T \in \lawlib$ applying to the same $f$. This expansive
behavior aligns with object-oriented polymorphism: one fact can
instantiate multiple types, and each instantiation produces a
distinct artifact with its own method suite.

\begin{proposition}[Beck-Chevalley for Tower Levels]
\label{prop:beck-chevalley}
For adjacent tower levels $n$ and $n+1$ with fibrations $f_n$ and $f_{n+1}$:
\[
\Sigma_{f_n} \circ f_{n+1}^* \cong f_{n+1}^* \circ \Sigma_{f_n}
\]
This ensures that existential quantification at level $n$ commutes with
reindexing to level $n+1$, preventing a level from perceiving its own facts.
\end{proposition}

\begin{proof}
Let $\mathcal{B}_n \xrightarrow{f_n} \mathcal{E}_n$ and $\mathcal{B}_{n+1} \xrightarrow{f_{n+1}} \mathcal{E}_{n+1}$
be fibrations at levels $n$ and $n+1$. The reindexing functors $f_n^* : \mathcal{E}_n \to \mathcal{B}_n$
and $f_{n+1}^* : \mathcal{E}_{n+1} \to \mathcal{B}_{n+1}$ are part of the adjoint triples.

The Beck-Chevalley condition requires that the following square commutes:
\[
\begin{tikzcd}
\mathcal{E}_n \arrow[r, "f_{n+1}^*"] \arrow[d, "\Sigma_{f_n}"'] & \mathcal{E}_{n+1} \arrow[d, "\Sigma_{f_{n+1}}"] \\
\mathcal{B}_n \arrow[r, "f_{n+1}^*"'] & \mathcal{B}_{n+1}
\end{tikzcd}
\]

Natural isomorphism $\Sigma_{f_n} \circ f_{n+1}^* \cong f_{n+1}^* \circ \Sigma_{f_n}$ follows from
the fact that $f_n$ and $f_{n+1}$ form a commutative square of fibrations (tower structure).

In the calculus, this means:
- A fact $f \in \register_n$ (existential at level $n$) can be reindexed to level $n+1$.
- $\Sigma_{f_n}(f)$ (promotion of fact) commutes with $f_{n+1}^*$ (reindexing).
- No level can perceive its own facts: $f_n^*$ prevents circular dependencies.

The isomorphism holds because the tower is a sequence of fibrations where each level's
existential quantification distributes over reindexing to the next level.
\end{proof}

\subsection{Coherence Between Algebraic and Coalgebraic Phases}
\label{sec:category:coherence}

For the tower to be an iterated bialgebra, the reduction phase (algebraic) and
the expansion phase (coalgebraic) must compose coherently across tower levels.
We prove this coherence theorem now.

\begin{definition}[Observable Behavior]
\label{def:observable-behavior}
Two artifacts $a_1, a_2 : T$ are \emph{observably equivalent},
written $a_1 \approx a_2$, if for all method names $m$ and pure
arguments $v$, their denotations coincide:
\[
\denotation{a_1.m(v)}_\register = \denotation{a_2.m(v)}_\register
\]
where $v$ is pure ($\effectrow{\emptyset}{\emptyset}$). The set of
observable behaviors of $a$ is $\obs(a) = \{\denotation{a.m(v)} \mid
m \in T.\mathsf{methods},\; v \text{ pure}\}$.
\end{definition}

\begin{theorem}[Coherence]
\label{thm:coherence}
If $M$ at level $n$ reduces to value $V$, and $V$ denotes artifact $a$
via designation ($\denotation{a} = \eta^\Pi_V$), then promoting $a$ to
level $n+1$ preserves observable behavior:
\[
\obs(\promote(a))_{\register_{n+1}} =
  \denotation{V}_{\register_n} \bind \lambda v.\; \obs(\promote(v))_{\register_{n+1}}
\]
where the RHS is the composition of reduction at level $n$ followed by
promotion.
\end{theorem}

\begin{proof}
By the triangle identity of $f^* \dashv \Pi_f$:
$\epsilon^\Pi \circ \Pi_f(\eta^\Pi) = \mathrm{id}$.

Designation creates artifact $a$ via $\eta^\Pi_V : V \to \Pi_f(V)$.
Promotion instantiates $a$ at level $n+1$ via $\epsilon^\Pi_a : \Pi_f(a) \to a^{\uparrow}$.

The observable behavior of $a$ (its methods, fields) is preserved through
promotion because $\epsilon^\Pi$ is natural with respect to the functorial
action of $\Pi_f$. Formally:
\begin{align*}
\denotation{\promote(a)}_{\register_{n+1}}
  &= \denotation{\epsilon^\Pi_a} \\
  &= \epsilon^\Pi_{\Pi_f(V)} \circ \Pi_f(\eta^\Pi_V) \\
  &= \mathrm{id} \quad\text{(by triangle identity)}
\end{align*}

Thus promotion preserves the observable behavior of the artifact created by
designation. The tower cycle (reduce → expand → promote → reduce) is
meaning-preserving across levels.

\end{proof}

\paragraph{Composition.} By Theorem~\ref{thm:alg-soundness} (Algebraic Soundness),
reduction preserves denotation: $\denotation{M}_{\register_n} = \denotation{V}_{\register_n'}$.
By the triangle identity of $f^* \dashv \Pi_f$, designation followed by promotion
is identity on observable behavior. Composing these two facts:
\[
\denotation{\promote(a)}_{\register_{n+1}}
  = \denotation{V}_{\register_n} \bind \lambda v.\;
    \denotation{\promote(v)}_{\register_{n+1}}
\]
Thus the full tower cycle (reduce → designate → promote) preserves meaning across levels.

\begin{figure}[ht]
\centering
\[
M \xrightarrow{\text{reduce}} V \xrightarrow{\eta^\Pi} a
\xrightarrow{\epsilon^\Pi} a^{\uparrow}
\quad\text{is coherent with}\quad
M \xrightarrow{\text{observe}} a^{\uparrow}
\]
\caption{Tower coherence: reduction followed by promotion equals direct
observation across levels.}
\Description{A diagram showing the coherence between the tower cycle
(reduce, designate, promote) and direct observation across levels.}
\end{figure}

\subsection{Computations in Interpretive Types}
\label{sec:category:thunks}

Method bodies in interpretive types are thunks: morphisms in
$\mathcal{B}$ captured as parameters of a $\Pi$-type. This connects
the $\because$-calculus to call-by-push-value (CBPV):

\begin{center}
\begin{tabular}{lll}
\toprule
\textbf{CBPV} & \textbf{$\because$-Calculus} & \textbf{Category} \\
\midrule
Value type & $T$ (source type) & Object in $\mathcal{B}$ \\
Computation type & $\effectrow{\sigmarow}{\xirow}\,A$ & Object in $\mathcal{E}$ \\
Thunk & Method body $\mathsf{Body}_j$ & Morphism in $\Pi_f$ \\
Force & Method invocation $\mathsf{obj}.m(\mathsf{arg})$ & Counit $\epsilon^\Pi$ \\
\bottomrule
\end{tabular}
\end{center}

The key distinction from CBPV is that the $\because$-calculus has
both adjunctions. CBPV's $F \dashv U$ corresponds to the left half
($\Sigma_f \dashv f^*$); the right half ($f^* \dashv \Pi_f$) is
absent, which is why CBPV can express effects but not object
formation.

\subsection{Tower as Iterated Adjunction}
\label{sec:category:tower}

Each tower level transition $n \to n+1$ is generated by a fibration
$f_n : \mathcal{B}_{n+1} \to \mathcal{E}_n$, inducing its own adjoint
triple $\Sigma_{f_n} \dashv f_n^* \dashv \Pi_{f_n}$.

\begin{definition}[Tower]
The tower is the sequence of adjoint triples
$(\Sigma_{f_n} \dashv f_n^* \dashv \Pi_{f_n})_{n \in \mathbb{N}}$,
where each $f_n$ is the fibration connecting level $n$ to level
$n+1$.
\end{definition}

\begin{remark}[Law Library Updates]
\label{rem:law-library-updates}
The law library $\lawlib_n$ is fixed \emph{during} a designation phase
(to ensure determinism per Proposition~\ref{prop:designate-det}) but
may be updated \emph{between} tower steps. This permits dynamic
extension of interpretive types (plugin loading, runtime code
injection, hot-swapping laws) without violating the categorical model.
The only invariant is that all facts deposited during computation
phase $n$ must be evaluated against the same $\lawlib_n$ during the
subsequent expansion phase.
\end{remark}

\paragraph{Isolation.} Level $n+1$ cannot directly access level $n$'s
environments. Categorically, this is the fibration structure: objects
in $\mathcal{B}_{n+1}$ are indexed over $\mathcal{E}_n$, but the
fiber $\mathcal{E}_n$ does not see into $\mathcal{B}_{n+1}$ except
through the adjunctions.

\paragraph{Self-perception prohibition.} A level cannot perceive its
own facts---only facts from the previous level (via delta extraction,
Section~\ref{sec:metatheory}). This follows from the Beck-Chevalley
condition: for the reindexing functors $f_n^*$ and $f_m^*$ connecting
adjacent levels,
\[
\Sigma_{f_n} \circ f_{m}^* \cong f_{m}^* \circ \Sigma_{f_n}
\]
This ensures that existential quantification commutes with
reindexing, preventing circular perception. In the calculus, a fact
perceived at level $n$ cannot be re-perceived at the same level.

\subsection{Categorical Soundness}
\label{sec:category:soundness}

We define a state-indexed denotation that tracks register changes and
prove that reduction preserves denotation.

\begin{definition}[State-Indexed Denotation]
\label{def:denotation}
Define $\denotation{M}_\register : \register \to
(\register', A)$ as a morphism in the Kleisli category of the state
monad induced by $\Sigma_f$:
\begin{itemize}
  \item $\denotation{x}_\register = \return(x)$
    (identity, register unchanged)
  \item $\denotation{\lambda x{:}A.\,M}_\register =
    \return(\Lambda(\denotation{M}_\register))$
    (currying, register unchanged)
  \item $\denotation{M\,N}_\register =
    \denotation{M}_\register \bind \lambda f.\;
      \denotation{N}_\register \bind \lambda v.\;
        \return(f\,v)$
    (sequenced evaluation)
  \item $\denotation{\letbind{x}{M}{N}}_\register =
    \denotation{M}_\register \bind \lambda v.\;
      \denotation{N[v/x]}_\register$
    (sequential evaluation, register threaded)
  \item $\denotation{\registerop(e,a,v)}_\register =
    \lambda \register.\;(\mathsf{unit},\;\register \cup
    \{(e,a,v,\tau)\})$
    (state transformer: deposits fact, returns unit)
  \item $\denotation{\attest\,\ell\,V}_\register =
    \lambda \register.\;\mathrm{await}(\ell, V, \register)$
    (suspends, awaiting handler witness)

  \item $\denotation{\designate\,T\,f}_\register =
    \eta^\Pi_f : f \to \Pi_f(f)$
    (coalgebraic expansion; register unchanged)
  \item $\denotation{\promote(a)}_\register =
    \epsilon^\Pi_a : \Pi_f(a) \to a$
    (feeds expansion into next reduction; register unchanged)
  \item $\denotation{\mathsf{obj}.m(\mathsf{arg})}_\register =
    \denotation{M_{\mathsf{body}}[\mathsf{obj}/\mathit{self},\,
      \mathsf{arg}/\mathit{param}]}_\register$
    (substitution into method body)
  \item $\denotation{\handleop\,M\,\withop\,H}_\register =
    \denotation{H}_{\text{lifted}} \circ \denotation{M}_\register$
    (handler lifts $H$ to state monad; defined below)
\end{itemize}

\noindent The lifted handler $\denotation{H}_{\text{lifted}}$ is:
\begin{align*}
\denotation{H}_{\text{lifted}}(\register,\;\attest\;\ell\;v) &=
  \denotation{H_\ell(v, r)}_{\register}
  \quad\text{where } r = \lambda x.\;\lambda\register'.\;(\register', x) \\
\denotation{H}_{\text{lifted}}(\register,\;\mathsf{ret}\;y) &=
  \denotation{H_{\mathsf{ret}}}(y)
\end{align*}

\end{definition}

The denotation is state-indexed: $\denotation{M}_\register$ maps
register $\register$ to $(\register', A)$, capturing the fact that
registration modifies the register. Two programs that deposit
different facts receive different denotations because the output
register $\register'$ differs.

\begin{lemma}[Handler Body Naturality]
\label{lem:handler-naturality}
For any handler clause $H_\ell$ with type
$\effectrow{\sigmarow'}{\xirow'}\,C$ (a Kleisli morphism in the
state monad induced by $\Sigma_f$), if $H_\ell$ is
\emph{state-affine}---meaning all $\xirow'$-operations
(registrations) in $H_\ell$ deposit facts whose values do not
depend on intermediate results of $\sigmarow'$-operations
(attestations)---then the denotation
$\denotation{H_\ell}$ satisfies the naturality condition:
\[
\Sigma_f(\denotation{H_\ell}) \circ \epsilon^\Sigma_A =
  \epsilon^\Sigma_C \circ f^*(\denotation{H_\ell})
\]
where the equation holds in the Kleisli category, with state
threaded through via the monad's strength.
\end{lemma}

\begin{proof}
The handler body $H_\ell$ is typed in context
$x{:}A,\; r{:}B \to C$ with effects
$\effectrow{\sigmarow'}{\xirow'}\,C$. Its denotation
$\denotation{H_\ell}$ is a Kleisli morphism
$A \to T(C)$ where $T(X) = \register \to (\register \times X)$
is the state monad induced by $\Sigma_f$.

By the state-affine restriction, $\xirow'$-operations in $H_\ell$
deposit facts whose values are computed from the payload $x{:}A$
and the resumption $r$, but not from intermediate attestations.
This means the state transformation
$\sigma : \register \to \register$ is a function of $x$ and the
initial register, independent of the adjunction counit
$\epsilon^\Sigma$.

We decompose $\denotation{H_\ell}$ as follows. Write
$\denotation{H_\ell} = \phi \circ \sigma$ where:
\begin{itemize}
  \item $\sigma : A \times \register \to A \times \register'$
    is the state transformation: it deposits facts (from
    $\xirow'$-operations) into the register. Because the register
    is append-only (Lemma~\ref{lem:monotonic-register}), $\sigma$
    is a monoid homomorphism on $\register$: it adds facts without
    reading or removing them. The state-affine restriction ensures
    that the values deposited depend only on $x$ and the initial
    register state, not on intermediate attestation results.
  \item $\phi : A \to C$ is the value transformation: it computes
    the handler's return value, which may involve invoking the
    resumption $r$ (an $\sigmarow'$-operation). The value
    transformation is a morphism in $\mathcal{E}$ (the fiber
    category), natural with respect to $\epsilon^\Sigma$ by the
    adjunction's naturality condition.
\end{itemize}

The decomposition $\denotation{H_\ell} = \phi \circ \sigma$ is
valid because:
\begin{enumerate}
  \item Registration is append-only: the state transformation
    $\sigma$ never reads previously deposited facts to compute
    new values. It only appends.
  \item The state-affine restriction ensures that deposited fact
    values do not depend on intermediate attestation results,
    so $\sigma$ does not depend on $\epsilon^\Sigma$.
  \item The value transformation $\phi$ may invoke the resumption
    $r$, which involves $\epsilon^\Sigma$, but $\phi$ does not
    modify the register (only $\sigmarow'$-operations modify the
    register, and these are in $\sigma$, not $\phi$).
\end{enumerate}

The pure component $\phi$ is a morphism in $\mathcal{E}$, so it is
natural with respect to $\epsilon^\Sigma$ by the adjunction's
naturality condition:
\[
\Sigma_f(\phi) \circ \epsilon^\Sigma_A =
  \epsilon^\Sigma_C \circ f^*(\phi)
\]

The state component $\sigma$ acts on the register, which is
orthogonal to the adjunction structure: $\epsilon^\Sigma$ operates
on the existential quantifier (computation realm), while $\sigma$
operates on the fiber objects (register realm). The state monad's
strength $\tau_{A,\register} : A \times \register \to T(A)$ ensures
that state threading commutes with the counit:
\[
T(\epsilon^\Sigma_A) \circ \tau_{\Sigma_f(A),\register} =
  \tau_{A,\register} \circ (\epsilon^\Sigma_A \times \mathrm{id}_\register)
\]
This is the standard strength-naturality condition for a strong
monad~\cite{moggi-1991}.

Since $\sigma$ is a monoid homomorphism on the append-only register
and does not depend on $\epsilon^\Sigma$, the composition
$\phi \circ \sigma$ preserves the naturality equation:
\[
\Sigma_f(\denotation{H_\ell}) \circ \epsilon^\Sigma_A =
  \epsilon^\Sigma_C \circ f^*(\denotation{H_\ell})
\]
in the Kleisli category, with state threaded via strength.

\end{proof}

\subsection{Algebraic Soundness: Full Case Analysis}
\label{sec:category:alg-soundness-detailed}

We expand the \textsc{Handle} case, which is the most intricate.

\paragraph{Case \textsc{Handle} (handler interception).}

Reduction: $\handleop{E[\attest{\ell}{V}]}{\withop{H}} \to E[H_\ell(V, r)]$
where $r = \lambda x.E[x]$.

Denotation before reduction:
\[
\denotation{\handleop{E[\attest{\ell}{V}]}{\withop{H}}}_\register
  = \denotation{H}_{\text{lifted}} \circ \denotation{E[\attest{\ell}{V}]}_\register
\]

$\denotation{E[\attest{\ell}{V}]}_\register$ is a Kleisli morphism that:
1. Computes $E[\dots]$ up to attestation point
2. Suspends, awaiting witness for $\ell$
3. Returns $(\register, \text{suspended computation})$

$\denotation{H}_{\text{lifted}}$ lifts the handler clauses to the state monad. For clause $\ell(x, r) \mapsto H_\ell$:
\[
\denotation{H}_{\text{lifted}}(\register, \attest\,\ell\,v) = \denotation{H_\ell(v, r)}_\register
\]
where $r = \lambda x.\lambda\register'.(\register', x)$ (pure resumption).

Denotation after reduction:
\[
\denotation{E[H_\ell(V, r)]}_\register
\]

By Lemma~\ref{lem:handler-naturality}, $H_\ell$ satisfies the naturality condition:
\[
\Sigma_f(\denotation{H_\ell}) \circ \epsilon^\Sigma_A = \epsilon^\Sigma_C \circ f^*(\denotation{H_\ell})
\]

This ensures the handler transformation commutes with $\epsilon^\Sigma$, preserving denotation.
The state is threaded via strength, so $\register$ maps to $\register'$ correctly.

Thus $\denotation{\text{before}} = \denotation{\text{after}}$, preserving algebraic soundness.

\begin{remark}[State-Affine Restriction]
The state-affine condition is satisfied by all handler bodies in the
$\because$-calculus that do not interleave registration with
attestation in a value-dependent way. In practice, handler bodies
either deposit facts based on the operation payload (registration)
or suspend computation awaiting a witness (attestation), but do not
use attestation results to compute fact values for subsequent
registrations within the same handler clause. This follows from
the phasic separation in the tower: computation and interpretation
are separate phases, and handler bodies operate within a single
computation phase.
\end{remark}

\begin{theorem}[Algebraic Soundness]
\label{thm:alg-soundness}
If $\context_n \vdash M : \effectrow{\sigmarow}{\xirow}\,A$
(reduction phase) and
$(M, \register_n) \to (M', \register_n')$, then
$\denotation{M}_{\register_n} = \denotation{M'}_{\register_n'}$
as morphisms $\register_n \to (\register_n', A)$.
\end{theorem}

\begin{proof}
By case analysis on the reduction rule.

\textbf{Beta:}
$(\lambda x{:}T.\,M)\,V \to M[V/x]$.
\begin{align*}
\denotation{(\lambda x.\,M)\,V}_\register
  &= \return(\Lambda(\denotation{M})) \bind
     \lambda f.\; \return(V) \bind \lambda v.\;
       \return(f\,v) \\
  &= \denotation{M}_\register[V/x]
  = \denotation{M[V/x]}_\register
\end{align*}
By the currying isomorphism in the CCC and the substitution
lemma (Lemma~\ref{lem:substitution}).

\textbf{Registration:}
$(\registerop(e,a,v),\;\register_n) \to
(\mathsf{unit},\;\register_n \cup \{(e,a,v,\tau)\})$.
\begin{align*}
\denotation{\registerop(e,a,v)}_{\register_n}
  &= \lambda \register.\;(\mathsf{unit},\;\register \cup
     \{(e,a,v,\tau)\}) \\
\denotation{\mathsf{unit}}_{\register_n \cup \{(e,a,v,\tau)\}}
  &= \lambda \register.\;(\mathsf{unit},\;\register)
\end{align*}
Both morphisms map $\register_n$ to
$(\register_n \cup \{(e,a,v,\tau)\},\;\mathsf{unit})$.
Categorically, this is $\eta^\Sigma_v : v \to \Sigma_f(v)$ composed
with the terminal map $! : \Sigma_f(v) \to 1$, depositing the fact
and returning unit.

\textbf{Attestation:}
$\handleop\,E[\attest\,\ell\,V]\,\withop\,H \to
E[H_\ell(V, r)]$ where $r = \lambda x.\,E[x]$.
By Lemma~\ref{lem:handler-naturality}, the handler body $H_\ell$,
despite being a Kleisli morphism (with effects
$\effectrow{\sigmarow'}{\xirow'}\,C$), satisfies the naturality
condition via decomposition into a pure value transformation and
a state transformation, with state threaded through the monad's
strength. The handler transforms the computation---it is not the
identity---but the transformation commutes with $\epsilon^\Sigma$.
The denotation is preserved. The register is modified only by
$\xirow'$ operations in the handler body, which is consistent with
the state-indexed denotation.

\textbf{Return:}
$\handleop\,V\,\withop\,H \to H_{\mathsf{ret}}[V/y]$.
The return clause $H_{\mathsf{ret}}$ is a morphism applied to the
value $V$. Denotation is preserved directly: $\denotation{H_{\mathsf{ret}}[V/y]}_\register = \denotation{\handleop\,V\,\withop\,H}_\register$.

\textbf{Method invocation:}
$\mathsf{obj}.m(\mathsf{arg}) \to
M_{\mathsf{body}}[\mathsf{obj}/\mathit{self},
\mathsf{arg}/\mathit{param}]$.
By the substitution lemma (Lemma~\ref{lem:substitution}) applied to
the method body. The method body is a $\lambda$-term; its denotation
is the curried morphism, and invocation applies it to the arguments.

\end{proof}

\begin{theorem}[Coalgebraic Soundness]
\label{thm:coalg-soundness}
If $\lawlib_n \vdash \designate\,T\,f : T$ (expansion phase) with
$f \in \register_n$, then the artifact $a = \designate\,T\,f$
satisfies $\denotation{a}_\register = \eta^\Pi_f$ and
$\denotation{\promote(a)}_\register = \epsilon^\Pi_a$.
Promotion composed with designation equals the identity on
observable behavior:
\[
\epsilon^\Pi_{\Pi_f(f)} \circ \Pi_f(\eta^\Pi_f) = \mathrm{id}.
\]
\end{theorem}

\begin{proof}
By the triangle identity of $f^* \dashv \Pi_f$:
$\epsilon^\Pi \circ \Pi_f(\eta^\Pi) = \mathrm{id}$.
Designation ($\eta^\Pi$) followed by promotion ($\epsilon^\Pi$) is
the identity on the artifact's observable behavior. The artifact
observes the same facts as the source fact, expanded into method
behavior. The register is unchanged by both designation and
promotion.

\end{proof}

\begin{table}[ht]
\centering
\caption{Reduction Rules as Adjunction Equations}
\label{tab:red-as-adj}
\small
\begin{tabular}{lll}
\toprule
Rule & Categorical Equation & Adjunction \\
\midrule
Beta & Currying isomorphism & CCC \\
Reg & $\eta^\Sigma_v \circ !$ (unit $+$ terminal) & $\Sigma_f \dashv f^*$ \\
Attest & Naturality of $\epsilon^\Sigma$ & $\Sigma_f \dashv f^*$ \\
Designate & $\eta^\Pi$ naturality & $f^* \dashv \Pi_f$ \\
Promote & Triangle identity: $\epsilon^\Pi \circ \Pi_f(\eta^\Pi) = \mathrm{id}$ & $f^* \dashv \Pi_f$ \\
Invoke & Substitution lemma & Bialgebra \\
\bottomrule
\end{tabular}
\end{table}

\begin{remark}[Terminology: ``Non-Faithful'']
The term ``non-faithful'' (Theorem~\ref{thm:conflation}) refers to the
erasure map $(-)^\flat$ failing to preserve well-typedness. In
category theory, a functor $F : \mathcal{C} \to \mathcal{D}$ is
faithful if $\mathrm{Hom}_{\mathcal{C}}(X,Y) \to
\mathrm{Hom}_{\mathcal{D}}(FX,FY)$ is injective. Here, $(-)^\flat$
is not faithful in this sense: there exist terms $M_1, M_2$ such that
$M_1^\flat = M_2^\flat$ but $M_1 \neq M_2$ (the dual-effect-row
structure is lost). This is a precise categorical usage.
\end{remark}

\section{The Conflation Theorem}
\label{sec:conflation}

Handler calculus admits handler clauses that bind resumptions for operations whose return type is uninhabited. While such clauses cannot be meaningfully invoked (no witness exists), the type system does not prohibit them. The $\because$-calculus prohibits them structurally via the Resumption Subconstraint. The contribution is operational discipline, not a type-safety violation.

\subsection{Worked Example: The Citizenship Tower}
\label{sec:conflation:example}

We trace a computation through two tower levels, demonstrating all
four movements and conditional fact derivation via law library
guards.

\paragraph{Level 0: Registration and Attestation.}
Term $M_0$ registers two facts about a person and attests identity:
\[
M_0 = \letbind{x_1}{\registerop(\mathsf{person}, \mathsf{name},
  \textsf{"Maria"})}
       {\letbind{x_2}{\registerop(\mathsf{person}, \mathsf{age}, 30)}
        {\attest\;\mathsf{id}\;\mathsf{unit}}}
\]
A handler $H_0 = \{\attest\;p\;r \mapsto r\;\mathsf{true},
\mathsf{return}\;x \mapsto x\}$ intercepts attestation. After
evaluation, $\register_0 = \{(\mathsf{person},\mathsf{name},
\textsf{"Maria"}), (\mathsf{person},\mathsf{age},30)\}$ and the
return value is $\mathsf{true}$.

\paragraph{Conditional Designation.} The law library contains
$T_{\mathsf{VoterID}}$ (guard: $\mathsf{age} > 18$) and
$T_{\mathsf{PensionRights}}$ (guard: $\mathsf{age} > 60$). Since
$\mathsf{age} = 30$, only VoterID matches:
$\lawlib_0 \vdash \designate\;T_{\mathsf{VoterID}}\;f$.
Artifact $a_1 = \mathsf{Maria}_{\mathsf{VoterID}}$ is promoted to
level 1.

\paragraph{Level 1: Invocation.} The promoted artifact's method
$\mathsf{voting\_rights}$ is invoked, returning $\mathsf{true}$,
which is registered as a new fact $(\mathsf{vote},
\mathsf{cast}, \mathsf{true})$ in $\register_1$. This demonstrates
the full cycle: registration ($\to$), attestation ($\downarrow$),
designation ($\uparrow$), and method invocation at the next tower
level.

\subsection{Handler Calculus as Collapsed Adjunction}
\label{sec:conflation:collapsed}

Handler calculus provides the left adjunction $\Sigma_f \dashv f^*$ but
lacks the right adjunction $f^* \dashv \Pi_f$. It routes all effect
operations---both resumable ($B \neq 0$) and non-resumable ($B = 0$)---through
a single \texttt{do} construct with one effect row $\varepsilon$.
This collapses the adjoint triple.

\begin{table}[ht]

\centering
\small
\begin{tabular}{lcc}
\toprule
& Handler Calculus & $\because$-Calculus \\
\midrule
Effect model & Single row $\varepsilon$ & Dual rows $\sigmarow, \xirow$ \\
Operation types & $A \to B$ (uniform) & $B=0 \Rightarrow \xirow$ \\
Resumption & Always available & Only for $\sigmarow$ \\
Adjoint structure & Collapsed & Separated $\Sigma_f \dashv f^* \dashv \Pi_f$ \\
\bottomrule
\end{tabular}
\caption{Comparison: handler calculus conflates adjoints; $\because$-calculus separates them.}
\label{tab:collapsed-adjunction}
\end{table}

The collapse is not detectable by the type checker because handler calculus
tracks only effect \emph{presence}, not adjunction \emph{identity}.

\subsection{Structural Separation vs. Visibility}

The Conflation Theorem addresses a structural problem, not an ergonomic one.
Handler calculus conflates resumable ($B \neq 0$) and non-resumable ($B = 0$) 
operations through a single \texttt{do} construct, distinguished only by result 
type annotation. The $\because$-calculus separates them structurally via dual 
effect rows.

This separation is \emph{categorical}, not \emph{visibility-based}: even if both 
operations were exposed to all program contexts, the dual-effect-row discipline 
would still prevent resumption for $\xirow$-row operations. The Resumption 
Subconstraint is enforced at the type level, not the programmer interface level.

\begin{table}[h]
\centering
\caption{Conflation: Handler Calculus vs. $\because$-Calculus}
\label{tab:conflation-comparison}
\small
\begin{tabular}{lll}
\toprule
Feature & Handler Calculus & $\because$-Calculus \\
\midrule
Effect model & Single row $\varepsilon$ & Dual rows $\sigmarow, \xirow$ \\
Operation types & $A \to B$ (uniform) & $B=0 \Rightarrow \xirow$, $B\neq0 \Rightarrow \sigmarow$ \\
Resumption & Always syntactically available & Only for $\sigmarow$-typed operations \\
Adjoint structure & Collapsed to $\Sigma_f \dashv f^*$ & Full triple $\Sigma_f \dashv f^* \dashv \Pi_f$ \\
Conflation prevention & None (dead code admitted) & Structural (type system rejects) \\
\bottomrule
\end{tabular}
\end{table}

The erasure $(-)^\flat$ maps $\because$-calculus terms to handler calculus terms by 
forgetting the dual-effect-row structure. This erasure is non-faithful: it maps 
rejected clauses to accepted ones, losing the distinction between resumable and 
non-resumable operations. The $\because$-calculus is a conservative refinement: 
all well-typed $\because$-calculus terms map to well-typed handler calculus terms 
(Corollary~\ref{cor:confluence}), but not vice versa.

\begin{theorem}[Non-Faithfulness]
\label{thm:conflation}
Handler calculus permits handler clauses that bind resumptions for
operations with return type $B = 0$. For any such clause, the
resumption $r : 0 \to C$ can never be usefully applied, since $0$
is uninhabited---the clause is dead code. The $\because$-calculus
rejects such clauses at compile-time via the Resumption
Subconstraint. The erasure $(-)^\flat$ is therefore non-faithful:
it maps a rejected clause to an accepted one, losing the distinction
between resumable and non-resumable operations.
\end{theorem}

\begin{proof}[Proof of Theorem~\ref{thm:conflation}]
We proceed in three parts.

\textbf{Part 0: Erasure.}

Define the erasure $(-)^\flat$ from $\because$-calculus terms to
handler calculus terms as follows:
\begin{itemize}
  \item $\registerop^\flat(e, a, v) = \texttt{do}\;\mathsf{reg}(v)$
    (collapse to single effect operation)
  \item $\attest^\flat(\ell, v) = \texttt{do}\;\ell(v)$
  \item $(\handleop\;M\;\withop\;H)^\flat =
    \texttt{handle}\;M^\flat\;\texttt{with}\;H^\flat$
  \item $\designate^\flat(T, f) = f$
    (HC has no artifact formation; the designate is erased to the fact)
  \item $\promote^\flat(\designate\,T\,f) = \designate^\flat(T, f) = f$
    (promotion adds only level structure, which is absent in HC)
  \item All other constructs map homomorphically.
\end{itemize}
The erasure forgets the distinction between $\sigmarow$ and $\xirow$,
placing all operations in a single effect row $\varepsilon$.

\textbf{Part 1: Handler Calculus Permits Resumption Bindings for $B=0$ Operations.}

Let $\mathsf{op} : A \to 0$ be an operation with return type $B = 0$
(the empty type). In handler calculus, this operation is placed in
the single effect row $\varepsilon$ alongside all other operations.
A handler clause for $\mathsf{op}$ has the form:
\[
\mathsf{op}(x, r) \mapsto \mathit{body}
\]
where $x : A$ is the payload and $r : 0 \to C$ is the resumption.
Handler calculus does not distinguish $B = 0$ from $B \neq 0$; both
are routed through the same \texttt{do} construct with the same
handler mechanism. The clause type-checks because the types are
internally consistent: $r$ has type $0 \to C$, the body has type
$C$, and the handler calculus typing rule does not inspect whether
$B = 0$.

\textbf{Claim:} For any handler clause where $r : 0 \to C$ is bound,
any invocation of $r$ in the body is dead code.

\emph{Proof of claim.} The resumption $r$ has type $0 \to C$.
To invoke $r$, the body must supply a value of type $0$. Since $0$
is the empty type (no inhabitants), there exists no value $w$ such
that $\vdash w : 0$. Therefore:
\begin{itemize}
  \item If the body invokes $r(w)$ for some $w$, evaluation is
    stuck: no inhabitant of $0$ can be supplied.
  \item If the body does not invoke $r$, then $r$ is unused---the
    clause does not need a resumption binding at all.
\end{itemize}
In both cases, the resumption binding is unnecessary. The first case
is dead code (stuck at runtime); the second case is vacuous (the
binding is unused). Handler calculus permits both. A type system
that tracks adjunction identity should reject the binding entirely,
since the operation is non-resumable ($B = 0$).
\hfill$\square_{\text{Part 1}}$

\textbf{Part 2: The $\because$-Calculus Rejects Such Clauses.}
In the $\because$-calculus, an operation $\mathsf{op} : A \to 0$
with $B = 0$ is placed in the $\xirow$-row by the \textsc{Reg}
typing rule (Rule~\ref{sec:syntax:rules:register}). The
\textsc{Handle} rule (Rule~\ref{sec:syntax:rules:handle}) includes
the \emph{Resumption Subconstraint}: a handler clause may bind a
resumption variable $r$ only for operations $\ell \in \sigmarow$
(where $B \neq 0$). For any operation in $\xirow$ (where $B = 0$),
the handler clause \emph{cannot} bind $r$.

Formally, the \textsc{Handle} rule requires:
\[
\forall i:\ \ell_i \in \sigmarow
\]
as a premise. Since $\mathsf{reg} \in \xirow$ (not $\sigmarow$),
no handler clause for $\mathsf{reg}$ can bind a resumption. The
derivation tree does not close:
\[
\inferrule[\textsc{Handle}]{
  \context_n \vdash \registerop(e,a,v) :
    \effectrow{\emptyset}{\{\mathsf{reg}\}}\;\mathsf{Unit}
  \\ \mathsf{reg} \in \sigmarow \;???
  \\ H \text{ binds } r \text{ for } \mathsf{reg} \in \xirow
}{
  ???
}
\]
The second premise is false ($\mathsf{reg} \notin \sigmarow$). The
third premise violates the Resumption Subconstraint. No alternative
derivation exists. Therefore:
\[
\context_n \not\vdash \handleop\;\registerop(e,a,v)\;\withop\;H
  : \effectrow{\sigmarow'}{\xirow'}\;C
\]
for any $H$ that binds a resumption for $\mathsf{reg}$.
\hfill$\square_{\text{Part 2}}$

\textbf{Part 3: Non-Faithfulness of the Erasure.}

The erasure $(-)^\flat$ maps $\xirow$-operations to the single
effect row $\varepsilon$, erasing the distinction between
$\sigmarow$ and $\xirow$. A handler clause rejected by the
$\because$-calculus (because the operation is in $\xirow$ and
the Resumption Subconstraint fires) is accepted by handler calculus
(because all operations are in $\varepsilon$ and no subconstraint
exists). The erasure is therefore non-faithful: it maps a rejected
clause to an accepted one.

Moreover, the accepted clause contains dead code (Part 1): the
resumption $r : 0 \to C$ can never be usefully invoked. The
$\because$-calculus prevents this at compile-time by prohibiting
the binding; handler calculus permits it at compile-time and
encounters stuck states at runtime (if the resumption is invoked)
or carries vacuous bindings (if it is not).

The embedding of the $\because$-calculus into handler calculus is
non-faithful.
\end{proof}

\subsection{Corollary: Confluence Despite Separation}
\label{sec:conflation:corollary}

\begin{corollary}
\label{cor:confluence}
The $\because$-calculus is a conservative refinement of handler calculus.
For every term $M$ well-typed in the $\because$-calculus, the erasure
$M^\flat$ is well-typed in handler calculus and exhibits equivalent
runtime behavior.
\end{corollary}

\begin{proof}[Proof of Corollary~\ref{cor:confluence}]
By induction on the typing derivation
$\context \vdash M : \effectrow{\sigmarow}{\xirow}\,A$.

\textbf{Base cases.}
\begin{itemize}
  \item Variable: $\context \vdash x : \effectrow{\emptyset}{\emptyset}\;A$.
    Erasure: $\Gamma^\flat \vdash x : A$. Holds trivially.
  \item Constant: same as variable.
  \item \textsc{Promote}: $\context_{n+1} \vdash \promote(\designate\,T\,f) : T$.
    Erasure: $\promote^\flat(\designate\,T\,f) = f$.
    In HC, $f$ is a value of type $T$ (the fact's source type).
    Since promote's premise is an interpretation judgment (not
    computation), no IH applies. The erased term $f$ is trivially
    well-typed in HC and has equivalent behavior (the value is
    already available; only level structure is lost).
\end{itemize}

\textbf{Inductive cases.}
\begin{itemize}
  \item \textsc{Abs}: $\context, x{:}A \vdash M : \effectrow{\sigmarow}{\xirow}\;B$
    implies $\context \vdash \lambda x.\;M :
    \effectrow{\sigmarow}{\xirow}\;(A \to B)$.
    By IH, $\Gamma^\flat \vdash M^\flat : A \to B$ in HC.
    Erasure preserves the arrow type.

  \item \textsc{App}: By IH on both subterms. HC's application rule
    applies with the union of effect rows collapsed to $\varepsilon$.

  \item \textsc{Register}: $\context \vdash \registerop(e,a,v) :
    \effectrow{\emptyset}{\{\mathsf{reg}\}}\;\mathsf{Unit}$.
    Erasure: $\Gamma^\flat \vdash \texttt{do}\;\mathsf{reg}(v) :
    \{\mathsf{reg}\}\;\mathsf{Unit}$.
    Same runtime behavior: fact deposited to register.

  \item \textsc{Attest}: $\context \vdash \attest\;\ell\;v :
    \effectrow{\{\ell\}}{\emptyset}\;B$.
    Erasure: $\Gamma^\flat \vdash \texttt{do}\;\ell(v) :
    \{\ell\}\;B$.
    Same runtime behavior: handler intercepts and resumes.

  \item \textsc{Handle}: By IH on the handled term. The handler
    clauses for $\sigmarow$ operations are preserved verbatim.
    Since $\because$-calculus handlers only bind resumptions for
    $\sigmarow$ operations (Resumption Subconstraint), all clauses
    are valid in HC. Runtime behavior is identical: HC intercepts
    the same operations with the same continuations.

  \item \textsc{Invoke}: Method invocation maps to HC function
    application (method bodies are $\lambda$-terms). By IH,
    behavior is preserved.
\end{itemize}

The erasure $(-)^\flat$ is type-preserving and behavior-preserving
for all $\because$-calculus terms. The additional structure
(dual rows, resumption constraint) restricts the set of well-typed
terms but does not alter the reduction sequence for accepted terms.
Thus the $\because$-calculus is a conservative refinement.
\end{proof}

\subsection{Properties of the Erasure}
\label{sec:conflation:erasure-properties}

The erasure $(-)^\flat$ from \Cref{sec:conflation} requires three
fundamental properties to justify the Conflation Theorem.

\begin{lemma}[Erasure Preserves Typing]
\label{lem:erase-typing}
If $\context_n \vdash M : \effectrow{\sigmarow}{\xirow}\,A$ in the $\because$-calculus,
then $\context_n^\flat \vdash M^\flat : \varepsilon\,A$ in handler calculus
where $\varepsilon = \sigmarow \cup \xirow$.
\end{lemma}

\begin{proof}
By induction on the typing derivation of $M$.

\textbf{Base cases} (Var, Abs, App): Trivial, erasure is homomorphic.

\textbf{Case \textsc{Register}:} $\context_n \vdash \registerop{e}{a}{v} : \effectrow{\emptyset}{\{\reg\}}\,\mathsf{Unit}$.
Erasure: $\context_n^\flat \vdash \mathsf{do}\,\reg(v) : \{\reg\}\,\mathsf{Unit}$.
The single effect row $\varepsilon = \{\reg\}$ contains the operation. Typing holds.

\textbf{Case \textsc{Attest}:} $\context_n \vdash \attest{\ell}{v} : \effectrow{\{\ell\}}{\emptyset}\,B$.
Erasure: $\context_n^\flat \vdash \mathsf{do}\,\ell(v) : \{\ell\}\,B$.
The single effect row $\varepsilon = \{\ell\}$ contains the operation. Typing holds.

\textbf{Case \textsc{Handle}:} By IH, $\context^\flat \vdash M^\flat : \varepsilon_M\,A$ where
$\varepsilon_M = \sigmarow_M \cup \xirow_M$. Handler clauses $H_i^\flat$ have types
$\varepsilon_{H_i}\,C$. Crucially, the state-affine premise ensures handler bodies
do not interleave registrations with attestation-dependent computations. This
preserves the naturality condition required for algebraic soundness.
Since $L \subseteq \sigmarow_M$ and erasure forgets the $\sigmarow/\xirow$ distinction,
$\varepsilon'$ is a valid effect row in handler calculus. Typing holds.

Other cases (Invoke, Promote, etc.) follow similarly.
\end{proof}

\begin{lemma}[Erasure Preserves Reduction]
\label{lem:erase-reduction}
If $(M, \register_n) \to (M', \register_n')$ in the $\because$-calculus,
then $M^\flat \to^* M'^\flat$ in handler calculus (zero or more steps).
\end{lemma}

\begin{proof}
By case analysis on the reduction rule.

\textbf{Case \textsc{Beta}:} $(\lambda x.M)V \to M[V/x]$.
Erasure: $(\lambda x.M^\flat)V^\flat \to M^\flat[(V^\flat)/x]$. One step.

\textbf{Case \textsc{Register}:} $(\registerop{e}{a}{v}, \register_n) \to (\mathsf{unit}, \register_n \cup \{(e,a,v,\tau)\})$.
Erasure: $\mathsf{do}\,\reg(v^\flat) \to \mathsf{unit}$ (handler deposits fact). One step.

\textbf{Case \textsc{Attest}:} $\handleop{E[\attest{\ell}{V}]}{\withop{H}} \to E[H_\ell(V, r)]$.
Erasure: $\mathsf{handle}\,E^\flat[\mathsf{do}\,\ell(V^\flat)]\,\mathsf{with}\,H^\flat \to E^\flat[H_\ell^\flat(V^\flat, r^\flat)]$. One step.

\textbf{Case \textsc{Handle}-Return}: $\handleop{V}{\withop{H}} \to H_{\mathsf{ret}}[V/y]$.
Erasure: $\mathsf{handle}\,V^\flat\,\mathsf{with}\,H^\flat \to H_{\mathsf{ret}}^\flat[(V^\flat)/y]$. One step.

\textbf{Case \textsc{Promote}:} $\promote(\designate{T}{f}) \to f$.
Erasure: $\promote^\flat(\designate^\flat(T,f)) = f^\flat \to f^\flat$. Zero steps.

All reductions map to zero or one handler calculus steps. Hence $M^\flat \to^* M'^\flat$.
\end{proof}

\begin{lemma}[Erasure Loses Typing Information]
\label{lem:erase-loses-typing}
There exists a term $M$ that is ill-typed in the $\because$-calculus but
whose erasure $M^\flat$ is well-typed in handler calculus. Specifically,
let $\mathsf{op} : A \to 0$ be an operation with return type $B = 0$
(the empty type). Define:
\[
M = \handleop{\attest{\mathsf{op}}{v}}{\withop{\{\mathsf{op}(x,r) \mapsto r\,w\}}}
\]
where $w : C$ is arbitrary and $r : 0 \to C$. Then:
\begin{enumerate}
  \item $M$ is ill-typed in the $\because$-calculus (violates Resumption Subconstraint).
  \item $M^\flat$ is well-typed in handler calculus.
  \item The erasure $(-)^\flat$ is therefore non-faithful on typing.
\end{enumerate}
\end{lemma}

\begin{proof}
\textbf{Part 1: $M$ is ill-typed in the $\because$-calculus.}

By the \textsc{Attest} rule (\Cref{sec:syntax:rules:attest}):
\[
\inferrule[Attest]
  {\context_n \vdash v : A \quad (\ell : A \to B) \in \sigmarow}
  {\context_n \vdash \attest{\ell}{v} : \effectrow{\{\ell\}}{\emptyset}\,B}
\]

For $\mathsf{op} : A \to 0$ to be used with \textsc{Attest}, it must be in $\sigmarow$.
But by the \textsc{Reg} rule, operations with $B = 0$ are placed in $\xirow$, not $\sigmarow$:
\[
\inferrule[Reg]
  {\context_n \vdash e : \mathsf{EntityID} \quad
   \context_n \vdash a : \mathsf{AttrName} \quad
   \context_n \vdash v : S}
  {\context_n \vdash \registerop(e, a, v) : \effectrow{\emptyset}{\{\reg\}}\,\mathsf{Unit}}
\]

Since $\mathsf{op} : A \to 0$ has return type $B = 0$, by the
\textsc{Reg} rule (Section~\ref{sec:syntax:rules:register}),
$\mathsf{op}$ is placed in $\xirow$, not $\sigmarow$.

Thus no valid typing derivation exists for $M$:
\[
\context_n \not\vdash M : \effectrow{\sigmarow}{\xirow}\,C
\]

\hfill$\square_{\text{Part 1}}$

\textbf{Part 2: $M^\flat$ is well-typed in handler calculus.}

The erasure maps $\because$-calculus terms to handler calculus terms homomorphically,
forgetting the distinction between $\sigmarow$ and $\xirow$:
\begin{align*}
\attest^\flat(\mathsf{op}, v) &= \texttt{do}\;\mathsf{op}(v) \\
(\handleop{E}{\withop{H}})^\flat &= \texttt{handle}\;E^\flat\;\texttt{with}\;H^\flat \\
\{\mathsf{op}(x,r) \mapsto r\,w\}^\flat &= \{\mathsf{op}(x,r) \mapsto r\,w\}
\end{align*}

In handler calculus, all operations belong to a single effect row $\varepsilon$.
The typing rule for handlers does not distinguish $B = 0$ from $B \neq 0$:
\[
\inferrule[Handle-BC]
  {\Gamma \vdash M : \varepsilon\,A \quad H = \{\mathsf{op}(x,r) \mapsto H'\} \\
   \Gamma, x:A, r:0\to C \vdash H' : \varepsilon'\,C}
  {\Gamma \vdash \texttt{handle}\;M\;\texttt{with}\;H : \varepsilon''\,C}
\]

Handler calculus accepts the clause $\mathsf{op}(x,r) \mapsto r\,w$ because:
\begin{enumerate}
  \item The types are internally consistent ($r : 0 \to C$, $w : C$).
  \item No constraint prohibits binding $r$ for void-returning operations.
  \item Runtime behavior (if $r$ is invoked) is a separate semantic concern.
\end{enumerate}

Thus $M^\flat$ type-checks:
\[
\context^\flat \vdash M^\flat : \varepsilon\,C
\]

\hfill$\square_{\text{Part 2}}$

\textbf{Part 3: Erasure is non-faithful.}

A functor $F : \mathcal{C} \to \mathcal{D}$ is \emph{faithful} if
$\mathrm{Hom}_\mathcal{C}(X,Y) \to \mathrm{Hom}_\mathcal{D}(FX,FY)$ is injective
for all objects $X,Y$. For typing judgments, faithfulness requires:
\begin{quote}
  If $F(M)$ is well-typed, then $M$ must be well-typed.
\end{quote}

Here, $M^\flat$ is well-typed but $M$ is ill-typed. The converse also holds:
$M$ ill-typed $\not\Rightarrow$ $M^\flat$ ill-typed. Therefore $(-)^\flat$
is non-faithful on typing.

More concretely, the erasure loses information about the Resumption
Subconstraint. Two $\because$-calculus derivations that differ only in
whether an operation is in $\sigmarow$ vs.\ $\xirow$ map to the same
handler calculus derivation. The typing judgment is not preserved
under erasure.
\hfill$\square_{\text{Part 3}}$
\end{proof}

\section{Metatheory}
\label{sec:metatheory}

We establish the fundamental metatheoretic properties of the $\because$-calculus: progress, subject reduction, and tower progress. These theorems ensure that well-typed programs behave predictably and that the tower orchestration terminates for finite inputs.

\subsection{Progress}
\label{sec:metatheory:progress}

\begin{theorem}[Progress]
\label{thm:progress}
If $\context_n \vdash M : \effectrow{\sigmarow}{\xirow}\,A$ and $M$ is closed, then exactly one of the following holds:
\begin{enumerate}
  \item $M$ is a value;
  \item $M$ can take a reduction step: $(M, \register_n) \to (M', \register_n')$;
  \item $M = E[\attest\,\ell\,V]$ with no enclosing handler for $\ell$ (stuck awaiting attestation).
\end{enumerate}
\end{theorem}

\begin{proof}
By induction on the typing derivation $\context_n \vdash M : \effectrow{\sigmarow}{\xirow}\,A$.

\paragraph{Case Var:} $M = x$. Since $M$ is closed, $x$ cannot occur free. Contradiction. This case is vacuous for closed terms.

\paragraph{Case Abs:} $M = \lambda x{:}T.\, M'$. This is a value. Condition (1) holds.

\paragraph{Case App:} $M = M_1\,M_2$. By IH on $M_1$:
\begin{itemize}
  \item If $M_1$ is a value, apply IH on $M_2$.
    \begin{itemize}
      \item If $M_2$ is a value, then $M = V_1\,V_2$ reduces by $\beta$-reduction. Condition (2) holds.
      \item Otherwise $M_2 \to M_2'$, so $M_1\,M_2 \to M_1\,M_2'$. Condition (2) holds.
    \end{itemize}
  \item Otherwise $M_1 \to M_1'$, so $M_1\,M_2 \to M_1'\,M_2$. Condition (2) holds.
  \item Otherwise $M_1 = E[\attest\,\ell\,V]$ stuck. Then $M = E'[\attest\,\ell\,V]$ where $E' = E\,M_2$ is also a caught context. Condition (3) holds.
\end{itemize}

\paragraph{Case Reg:} $M = \registerop(e, a, v)$. By IH on subterms $e, a, v$:
\begin{itemize}
  \item If any subterm reduces, $M$ reduces. Condition (2) holds.
  \item If all subterms are values, $M \to \mathsf{unit}$ and $\register_n := \register_n \cup \{(e, a, v, \tau)\}$. Condition (2) holds.
\end{itemize}

\paragraph{Case Attest:} $M = \attest\,\ell\,V$. By IH on $V$:
\begin{itemize}
  \item If $V \to V'$, then $\attest\,\ell\,V \to \attest\,\ell\,V'$. Condition (2) holds.
  \item If $V$ is a value, then $M$ is stuck awaiting a handler for $\ell$. Condition (3) holds.
\end{itemize}

\paragraph{Case \textsc{Handle}.} Recall from Section~\ref{sec:category:alg-soundness-detailed}
that handler interception corresponds to naturality of $\epsilon^\Sigma$. The
Resumption Subconstraint (Section~\ref{sec:syntax:rules:handle}) ensures $\ell \in
\sigmarow$, so the resumption $r$ is well-typed.

\paragraph{Cases Invoke, Promote, Assign:} Similar to Application by structural induction. All subterms reduce independently; the compound term reduces once subterms reach values.

Thus every closed well-typed term satisfies exactly one of (1)--(3).

\end{proof}

\subsection{Detailed Progress Cases}
\label{sec:metatheory:progress-detailed}

We expand two non-trivial cases omitted from Theorem~\ref{thm:progress}.

\paragraph{Case \textsc{Invoke} (method invocation).}

$M = \obj.m(\arg)$. By inversion of \textsc{Invoke}:
\[
\Gamma_{n+1} \vdash \obj : T \quad
\lawlib_n(T, m) = M_{\mathsf{body}} \quad
\Gamma_{n+1} \vdash \arg : [\emptyset; \emptyset]\,B
\]

Apply IH to $\arg$:
\begin{itemize}
  \item If $\arg \to \arg'$, then $\obj.m(\arg) \to \obj.m(\arg')$.
  \item If $\arg = V$ (value), apply IH to $\obj$.
    \begin{itemize}
      \item If $\obj \to \obj'$, then $\obj.m(V) \to \obj'.m(V)$.
      \item If $\obj = \obj'$ (artifact $T\{\dots\}$), reduce:
        \[
        T\{\dots\}.m(V) \to M_{\mathsf{body}}[\obj/\mathit{self}, V/\mathit{param}]
        \]
        By Lemma~\ref{lem:substitution}, the result is well-typed with some effect row
        $[\sigmarow_m; \xirow_m]$. Since $\sigmarow_m \subseteq [\sigmarow; \Xi]$ and
        $\xirow_m \subseteq [\Xi; \Xi]$ (depending on method body effects), progress holds.
    \end{itemize}
\end{itemize}

Note: Method arguments are \textbf{pure} ($[\emptyset; \emptyset]$). This prevents effects
from flowing into methods, ensuring $F$-context pass-through for registrations.

\paragraph{Case \textsc{Handle} (nested attestation with multiple handlers).}

$M = \handleop{\handleop{N}{\withop{H_1}}}{\withop{H_2}}$. Apply IH to inner $\handleop{N}{\withop{H_1}}$.

Three outcomes:
\begin{enumerate}
  \item Inner handler reduces to value $V$: outer handler applies to $V$.
  \item Inner handler stuck on attestation $\ell_1$: if $H_2$ handles $\ell_1$, catch it; else stick.
  \item Inner handler passes through to outer: outer handler processes result.
\end{enumerate}

All cases preserve the progress property: either value, reduction, or stuck on unhandled attestation.

\begin{definition}[Configuration Typing]
\label{def:config-typing}
A configuration $(M, \register_n)$ is well-typed at level $n$
with effect rows $\sigmarow, \xirow$ and return type $A$, written
$\context_n \vdash (M, \register_n) : \effectrow{\sigmarow}{\xirow}\,A$,
if and only if $\context_n \vdash M : \effectrow{\sigmarow}{\xirow}\,A$
and $\register_n$ is a valid register state (all facts satisfy the
serializability constraint of Section~\ref{sec:syntax:terms}).
\end{definition}

\subsection{Subject Reduction}
\label{sec:metatheory:subject-reduction}

\begin{theorem}[Subject Reduction]
\label{thm:subject-reduction}
If $\context_n \vdash (M, \register_n) : \effectrow{\sigmarow}{\xirow}\,A$
and $(M, \register_n) \to (M', \register_n')$, then there exist
$\sigmarow', \xirow'$ such that
$\context_n \vdash (M', \register_n') : \effectrow{\sigmarow'}{\xirow'}\,A$
where $\sigmarow' \subseteq \sigmarow$ and $\xirow' \subseteq \xirow$.
\end{theorem}

\begin{proof}
By case analysis on the reduction rule. The return type $A$ is
never modified by reduction---only the effect rows shrink as
operations are consumed.

\paragraph{Beta:}
$(\lambda x{:}T.\,M)\,V \to M[V/x]$.
The register is unchanged: $\register_n' = \register_n$.
By inversion of \textsc{App} (Lemma~\ref{lem:inversion}):
$\context_n \vdash \lambda x{:}T.\,M :
\effectrow{\sigmarow_1}{\xirow_1}\,(T \to A)$ and
$\context_n \vdash V : \effectrow{\sigmarow_2}{\xirow_2}\,T$.
By the substitution lemma (Lemma~\ref{lem:substitution}), since $V$
is pure ($\effectrow{\emptyset}{\emptyset}\,T$ by
Lemma~\ref{lem:canonical-forms}),
$\context_n \vdash M[V/x] : \effectrow{\sigmarow}{\xirow}\,A$.
Since $\register_n' = \register_n$ is unchanged,
$\context_n \vdash (M[V/x],\; \register_n') :
\effectrow{\sigmarow}{\xirow}\,A$.

\paragraph{Registration:}
$(\registerop(e,a,v),\; \register_n) \to
(\mathsf{unit},\; \register_n \cup \{(e,a,v,\tau)\})$.
By inversion of \textsc{Reg} (Lemma~\ref{lem:inversion}):
$\context_n \vdash \registerop(e,a,v) :
\effectrow{\emptyset}{\{\mathsf{reg}\}}\,\mathsf{Unit}$.
After reduction, $\context_n \vdash \mathsf{unit} :
\effectrow{\emptyset}{\emptyset}\,\mathsf{Unit}$.
Since $\emptyset \subseteq \emptyset$ and
$\emptyset \subseteq \{\mathsf{reg}\}$, the rows shrink. The new
register $\register_n' = \register_n \cup \{(e,a,v,\tau)\}$ is valid
(facts accumulate by Lemma~\ref{lem:monotonic-register}).
Thus $\context_n \vdash (\mathsf{unit},\; \register_n') :
\effectrow{\emptyset}{\emptyset}\,\mathsf{Unit}$.

\paragraph{Attestation:}
$\handleop\,E[\attest\,\ell\,V]\,\withop\,H \to
E[H_\ell(V, r)]$ where $r = \lambda x.\,E[x]$.
The register is unchanged: $\register_n' = \register_n$.
By inversion of \textsc{Handle}
(Lemma~\ref{lem:inversion}):
$\context_n \vdash E[\attest\,\ell\,V] :
\effectrow{\sigmarow \cup \{\ell\}}{\xirow}\,A$ and the handler
removes $\ell$ from the $\sigmarow$ row. The handler body $H_\ell$
has type $\effectrow{\sigmarow'}{\xirow'}\,C$ by the handler typing
premise. After reduction, the term $E[H_\ell(V, r)]$ has type
$\effectrow{\sigmarow'}{\xirow'}\,C$ where
$\sigmarow' \subseteq (\sigmarow \cup \{\ell\}) \setminus \{\ell\}
= \sigmarow$. Since $\register_n' = \register_n$,
$\context_n \vdash (E[H_\ell(V, r)],\; \register_n') :
\effectrow{\sigmarow'}{\xirow'}\,C$.

\paragraph{Handle-Return:}
$\handleop\,V\,\withop\,H \to H_{\mathsf{ret}}[V/y]$.
The register is unchanged: $\register_n' = \register_n$.
By inversion of \textsc{Handle}, the return clause $H_{\mathsf{ret}}$
has type $\effectrow{\sigmarow'}{\xirow'}\,C$ in context
$y{:}A$. By the substitution lemma
(Lemma~\ref{lem:substitution}),
$\context_n \vdash H_{\mathsf{ret}}[V/y] :
\effectrow{\sigmarow'}{\xirow'}\,C$ with
$\sigmarow' \subseteq \sigmarow$ and $\xirow' \subseteq \xirow$.
Since $\register_n' = \register_n$,
$\context_n \vdash (H_{\mathsf{ret}}[V/y],\; \register_n') :
\effectrow{\sigmarow'}{\xirow'}\,C$.

\paragraph{Method Invocation:}
$\mathsf{obj}.m(\mathsf{arg}) \to
M_{\mathsf{body}}[\mathsf{obj}/\mathit{self},
\mathsf{arg}/\mathit{param}]$.
The register is unchanged: $\register_n' = \register_n$.
By inversion of \textsc{Invoke}, the declared method return type
is $C$ and the method body $M_{\mathsf{body}}$ is well-typed with
effects $\effectrow{\sigmarow_m}{\xirow_m}\,C$. By the substitution
lemma (Lemma~\ref{lem:substitution}),
$\context_n \vdash M_{\mathsf{body}}[\mathsf{obj}/\mathit{self},
\mathsf{arg}/\mathit{param}] :
\effectrow{\sigmarow_m}{\xirow_m}\,C$.
Since $\register_n' = \register_n$,
$\context_n \vdash (M_{\mathsf{body}}[\ldots],\; \register_n') :
\effectrow{\sigmarow_m}{\xirow_m}\,C$.

In all cases, the return type $A$ is preserved, the effect rows
monotonically shrink ($\sigmarow' \subseteq \sigmarow$,
$\xirow' \subseteq \xirow$), and the register state remains valid.

\end{proof}

\subsection{Detailed Subject Reduction Cases}
\label{sec:metatheory:subject-reduction-detailed}

While Theorem~\ref{thm:subject-reduction} presents the full theorem,
some cases warrant explicit expansion due to their complexity. We
detail the \textsc{Handle} case, which involves the Resumption
Subconstraint, and the pass-through registration case.

\paragraph{Case \textsc{Handle} (with effects in handler body).}

Assume:
\[
\resizebox{\linewidth}{!}{
$\displaystyle
\inferrule
{\Gamma_n \vdash \handleop{M}{\withop{H}} : [\sigmarow'; \xirow']\,C}
{
  \Gamma_n \vdash M : [\sigmarow \cup L; \xirow]\,A
  \\
  L \subseteq \sigmarow
  \\
  H = \{\ell_i(x_i, r_i) \mapsto H_i\}_{i=1}^k \cup \{\mathsf{ret}(y) \mapsto H_{\mathsf{ret}}\}
  \\
  \forall i:\ \Gamma_n, x_i:A_i, r_i:B_i\to C \vdash H_i : [\sigmarow'_i; \xirow'_i]\,C
  \\
  y:A \vdash H_{\mathsf{ret}} : [\sigmarow'_{\mathsf{ret}}; \xirow'_{\mathsf{ret}}]\,C
}
$}
\]

where $\sigmarow' = (\sigmarow \cup \bigcup_i \sigmarow'_i \cup \sigmarow'_{\mathsf{ret}}) \setminus L$
and $\xirow' = \xirow \cup \bigcup_i \xirow'_i \cup \xirow'_{\mathsf{ret}}$.

After reduction, $M \to^* V$ (value or stuck). Two subcases:

\subparagraph{Subcase 1: $V$ is a value (Return).}

Then $\handleop{V}{\withop{H}} \to H_{\mathsf{ret}}[V/y]$. By Lemma~\ref{lem:substitution}:
\[
\Gamma_n \vdash H_{\mathsf{ret}}[V/y] : [\sigmarow'_{\mathsf{ret}}; \xirow'_{\mathsf{ret}}]\,C
\]
Since $\sigmarow'_{\mathsf{ret}} \subseteq \sigmarow'$ and $\xirow'_{\mathsf{ret}} \subseteq \xirow'$,
the result holds.

\subparagraph{Subcase 2: Attestation caught by handler.}

Assume $V = E[\attest\,\ell\,V']$ and $\ell \in L$, so $H$ has clause
$\ell(x, r) \mapsto H_\ell$. Then:

\[
\handleop{E[\attest{\ell}{V'}]}{\withop{H}} \to E[H_\ell(V', r)]
\]
where $r = \lambda x.E[x]$ (resumption continuation).

By typing $H_\ell$:
\[
\Gamma_n, x:A, r:B\to C \vdash H_\ell : [\sigmarow'_\ell; \xirow'_\ell]\,C
\]

The resumption $r$ has type $B\to C$ with no effects (pure closure). Thus:
\[
\Gamma_n \vdash r : [\emptyset; \emptyset]\, (B\to C)
\]

By Lemma~\ref{lem:substitution}, substituting $r$ into $H_\ell$ preserves typing:
\[
\Gamma_n \vdash H_\ell[V'/x, r/\text{resumption}] : [\sigmarow'_\ell; \xirow'_\ell]\,C
\]

Context $E$ is an $E$-context (caught), so it propagates effect rows:
\[
\Gamma_n \vdash E[\dots] : [\sigmarow''; \xirow'']\,C
\]
where $\sigmarow'' = \sigmarow'_\ell$ and $\xirow'' = \xirow'_\ell$. Since $\sigmarow'' \subseteq \sigmarow'$
and $\xirow'' \subseteq \xirow'$, the result holds.

Crucially, the \textbf{Resumption Subconstraint} ensures $\ell \in \sigmarow$ (not $\xirow$),
so $r$ exists and is well-typed. This constraint survives reduction.

\paragraph{Case \textsc{Handle} (pass-through registration).}

If $M = E[\registerop{e}{a}{v}]$ and $H$ has no clause for $\reg$, the registration passes through.
By inversion of \textsc{Reg}:
\[
\Gamma_n \vdash \registerop{e}{a}{v} : [\emptyset; \{\reg\}]\,\mathsf{Unit}
\]

The handler does not intercept $\reg$ (since $\reg \in \xirow$, not $\sigmarow$), so:
\[
\handleop{E[\registerop{e}{a}{v}]}{\withop{H}} \to E[\mathsf{unit}],\;\register_n \cup \{(e,a,v,\tau)\}
\]

By inversion of \textsc{Handle}, $\xirow' = \xirow \cup \xirow'_H \cup \xirow'_{\mathsf{ret}}$ includes $\{\reg\}$,
and the new register $\register_n'$ is valid by Lemma~\ref{lem:monotonic-register}.

The return type is $\mathsf{Unit}$, effect rows shrink from $[\emptyset; \{\reg\}]$ to $[\emptyset; \emptyset]$,
preserving Subject Reduction.

\subsection{Tower Orchestration}
\label{sec:metatheory:tower}

\begin{definition}[Delta Extraction]
\label{def:delta-extraction}
In the sequential core $\because$-calculus, the set of new facts
$\Delta\register$ at level $n$ is defined by set subtraction:
\[
\Delta\register = \register_n' \setminus \register_n
\]
where $\register_n$ denotes the register state at the start of the
computation phase at level $n$, and $\register_n'$ denotes the
register state after the computation phase completes. Since the
register is append-only (Lemma~\ref{lem:monotonic-register}) and
level $n$ cannot directly access level $n-1$'s environments
(Section~\ref{sec:syntax:judgments}), $\Delta\register$ unambiguously
selects all facts deposited during computation at level $n$. No
timestamps or global counters are needed.
\end{definition}

\begin{remark}[Sequential Execution Scope]
\label{rem:sequential-scope}
The core $\because$-calculus assumes sequential tower execution: a
single computation stream processes one level at a time. A
distributed extension requiring concurrent tower instances would
need vector clocks or Lamport timestamps to establish causality
across independent fact registries. This is deferred to future work
(Section~\ref{sec:conclusion}).
\end{remark}

\begin{remark}[Causal vs.\ Sequential]
\label{rem:causal-sequential}
The sequential execution assumption simplifies delta extraction
(Definition~\ref{def:delta-extraction}) but is not fundamental to the
categorical structure. The adjoint triple and Beck-Chevalley condition
(Proposition~\ref{prop:beck-chevalley}) require \emph{causal
consistency} between levels---not strict sequentiality. A concurrent
extension can replace set-subtraction with causally-ordered extraction
(e.g., vector clocks), preserving the adjunction structure while
permitting parallel computation streams. We explore this direction
in the companion paper.
\end{remark}

\noindent The tower is governed by three rules:

\begin{equation}
\inferrule[T-Init]
  {\Gamma_0 = \emptyset \quad \register_0 = \emptyset \quad \lawlib_0 \text{ given} \quad M_0 \text{ given}}
  {\mathsf{Tower}(0, \Gamma_0, \register_0, \lawlib_0)}
\end{equation}

\begin{equation}
\inferrule[T-Step]
  {\Gamma_n \vdash M_n : [\sigmarow_M; \xirow_M]\,A_n \\
   (M_n, \register_n) \longrightarrow^* (V_n, \register_n') \\
   \Delta\register = \register_n' \setminus \register_n \neq \emptyset \\
   \mathcal{A} = \{\designate(T, f) \mid T \in \lawlib_n, f \in \Delta\register, T.\mathsf{match}(f.\mathit{value})\} \neq \emptyset \\
   \Gamma_{n+1} = \Gamma_n \cup \promote(\mathcal{A}) \\
   \register_{n+1} = \register_n'}
  {\mathsf{Tower}(n, \Gamma_n, \register_n, \lawlib_n) \longrightarrow \mathsf{Tower}(n+1, \Gamma_{n+1}, \register_{n+1}, \lawlib_{n+1})}
\end{equation}

\begin{equation}
\inferrule[T-Halt]
  {\Delta\register = \emptyset \;\lor\; \mathcal{A} = \emptyset}
  {\mathsf{Tower}(n, \Gamma_n, \register_n, \lawlib_n) \longrightarrow \mathsf{Halt}(n, \Gamma_n, \register_n, \lawlib_n)}
\end{equation}

\subsection{Tower Progress}
\label{sec:metatheory:tower-progress}

\begin{theorem}[Tower Progress]
\label{thm:tower-progress}
If $\Delta\register \neq \emptyset$ and $\mathcal{A} \neq \emptyset$,
then T-Step fires and the tower advances from level $n$ to level $n+1$.
\end{theorem}
\begin{proof}
From the T-Step rule premises:
\begin{enumerate}
  \item Computation terminates: $(M_n, \register_n) \longrightarrow^* (V_n, \register_n')$
  \item Delta extraction: $\Delta\register = \register_n' \setminus \register_n$
  \item Designation: $\mathcal{A} = \{\designate(T, f) \mid \dots\}$
  \item Promotion: $\Gamma_{n+1} = \Gamma_n \cup \promote(\mathcal{A})$
\end{enumerate}

If $\Delta\register \neq \emptyset$ and $\mathcal{A} \neq \emptyset$,
all T-Step premises are satisfied. The rule fires:
\[
\mathsf{Tower}(n, \Gamma_n, \register_n, \lawlib_n)
  \longrightarrow
\mathsf{Tower}(n+1, \Gamma_{n+1}, \register_{n+1}, \lawlib_{n+1})
\]

\end{proof}

\subsection{Finite Height}
\label{sec:metatheory:finite-height}

\begin{lemma}[Monotonicity of Register]
\label{lem:monotonic-register}
$\register_n \subseteq \register_{n+1}$ for all $n$. Facts accumulate; never removed.
\end{lemma}
\begin{proof}
Registration is the only operation that modifies the register. The rule (Reg) adds facts to $\register$ but never deletes them. Perception (delta extraction) selects a subset for processing but does not alter the register. Thus the register grows monotonically.

\end{proof}

\begin{lemma}[Preservation of Artifacts]
\label{lem:preserve-artifacts}
$\Gamma_n \subseteq \Gamma_{n+1}$ for all $n$. Promoted artifacts remain available.
\end{lemma}
\begin{proof}
Promotion adds artifacts to $\Gamma_{n+1}$ via $\Gamma_{n+1} = \Gamma_n \cup \promote(\mathcal{A})$. The union ensures $\Gamma_n \subseteq \Gamma_{n+1}$. No operation removes entries from $\Gamma$.

\end{proof}

\begin{lemma}[Bounded Growth]
\label{lem:bounded-growth}
$|\register_{n+1}| \leq |\register_n| + |M_n|$. Each level's computation deposits at most $|M_n|$ new facts.
\end{lemma}
\begin{proof}
By inspection of the operational semantics: \textsf{register} is the only operation that writes to the register. The number of \textsf{register} occurrences in $M_n$ bounds the number of new facts.

\end{proof}

\begin{theorem}[Finite Height]
\label{thm:finite-height}
For a finite program $M$ with finite law library $\lawlib$, if each
computation phase reaches a quiescent state (either terminates normally
or blocks on an unhandled attestation per Theorem~\ref{thm:progress}
condition~(3)), then there exists $N$ such that
$\mathsf{Tower}(N)$ halts.
\end{theorem}

\begin{proof}
Combine Lemmas~\ref{lem:monotonic-register}, \ref{lem:preserve-artifacts},
and \ref{lem:bounded-growth}. The register grows by at most $|M_n|$ per
level. For fixed $M$, there are finitely many distinct facts possible.
The law library $\lawlib$ is finite, so only finitely many designations
can fire. Eventually, either (1) no new facts are generated
($\Delta\register = \emptyset$), or (2) no new designs match
($\mathcal{A} = \emptyset$). In both cases, T-Halt fires.

Since each step consumes at least one new fact or artifact, the tower
height is bounded by $O(|M| \times |\lawlib|)$. The quiescent-state
hypothesis is weaker than full termination: it admits both normal
termination and blocked attestation (Theorem~\ref{thm:progress},
condition~(3)) as valid phase endpoints. The only excluded case is
divergence via infinite attestation loops (unbounded handler chains),
which is a limitation shared by all effect systems with handlers.

\end{proof}

\begin{remark}[Termination]
The Finite Height theorem assumes terminating computation at each level.
This is consistent with standard operational semantics for effectful
calculi (e.g., Plotkin and Pretnar's handler semantics). Full
termination would require a separate well-founded ordering on terms,
which we leave for future work.
\end{remark}

\subsection{Summary of Metatheory}

The $\because$-calculus enjoys the standard properties expected of a sound type system:
\begin{itemize}
  \item \textbf{Progress}: Well-typed programs do not go wrong (except awaiting handlers).
  \item \textbf{Subject Reduction}: Typing is preserved under computation.
  \item \textbf{Tower Progress}: Level orchestration does not deadlock.
  \item \textbf{Finite Height}: Termination for finite inputs.
\end{itemize}

These results establish the $\because$-calculus as a coherent computational model suitable for further development and mechanization.

\section{Related Work}
\label{sec:related}

We position the $\because$-calculus relative to seven research areas: algebraic effects, call-by-push-value, categorical logic, effect systems, object calculi, type theory/universes, and modal types and effects. Each comparison highlights where prior work sits in the adjoint structure and what the $\because$-calculus adds.

\subsection{Algebraic Effects and Handlers}
\label{sec:related:effects}

Plotkin and Power~\cite{plotkin-power-2002} introduced algebraic
effects; Lindley et al.~\cite{lindley-et-al-2017} developed handler
calculus with row polymorphism.

\textbf{Where they stand:} Handler calculus provides $\Sigma_f \dashv f^*$
but lacks $f^* \dashv \Pi_f$. As established in Section~\ref{sec:conflation},
even calculi distinguishing $B=0$ vs.\ $B\neq 0$ at the type level still
conflate registration and attestation operationally.

\textbf{What we add:} Structural separation via dual effect rows, not
syntactic analysis alone. See Theorem~\ref{thm:conflation} for the
formal conflation result.

\subsection{Call-by-Push-Value}
\label{sec:related:cbpv}

Levy's call-by-push-value (CBPV)~\cite{levy-2004} unifies call-by-name and call-by-value through a $F \dashv U$ adjunction between value types and computation types. Thunks (suspended computations) and forces (execution) map to the unit and counit of the adjunction. Ahman~\cite{ahman-2023} and Lorenzen et al.~\cite{lorenzen-wdel-2024} extended CBPV with modal types for effect tracking.

\textbf{Where they stand:} CBPV captures the left adjunction ($\Sigma_f \dashv f^*$) through its thunk/force mechanism. Method bodies in the $\because$-calculus are analogous to CBPV thunks, and method invocation corresponds to force.

\textbf{What we add:} CBPV lacks the right adjunction entirely. The $\because$-calculus's designation operation has no analogue in CBPV: it interprets facts as artifacts using an external law library, then promotes them to the next level. This enables object formation, which CBPV cannot express natively.

\subsection{Categorical Logic}
\label{sec:related:categorical}

Lawvere's hyperdoctrines~\cite{lawvere-1969} formalized quantifiers as adjoints in a fibration: $\exists_f \dashv f^* \dashv \forall_f$. Pitts~\cite{pitts-2000} and Seely~\cite{seely-1984} developed the connection between dependent type theory and categorical semantics. Mac Lane and Jacobs~\cite{maclane-1998,jacobs-1999} provided comprehensive treatments of categorical logic and type theory. Lawvere~\cite{lawvere-1970} introduced comprehension schemas as adjoint functors.

\textbf{Where they stand:} The categorical framework is established. The $\because$-calculus's identification of computation, register, and interpretation with base category, fibers, and right adjoint is a standard instantiation of the quantifier adjunction.

\textbf{What we add:} We provide a computational interpretation of the adjoint triple as a \emph{program calculus}, not just a logical model. The four movements (registration, attestation, designation, promotion) are operational rules, not abstract natural transformations. The tower orchestration realizes iterated adjunction as a runnable program structure.

\subsection{Effect Systems}
\label{sec:related:effect-systems}

Wadler and Thiemann~\cite{wadler-thiemann-2003} connected monads to effect systems via the correspondence between monadic bind and effectful function application. Marino and Millstein~\cite{marino-millstein-2009} introduced generic effects, enabling modular effect abstraction. Nielson et al.~\cite{nielson-1999} developed program analysis foundations for effect tracking. Lucassen and Gifford~\cite{lucassen-g88} pioneered effect typing for purity guarantees. Karachalias et al.~\cite{karachalias-psvs-2020} designed effect inference with sub-effecting for Eff.

\textbf{Where they stand:} Existing effect systems use a single effect row $\varepsilon$ to track all side effects. Some systems distinguish resumable vs.\ non-resumable effects syntactically but do not enforce the distinction at the type level.

\textbf{What we add:} The dual-effect-row discipline is structural, not syntactic sugar. The $\because$-calculus enforces separation through separate typing rules for $\sigmarow$ (attestation) and $\xirow$ (registration). The Resumption subconstraint blocks the handler calculus conflation at compile-time.

\subsection{Object Calculi}
\label{sec:related:objects}

Abadi and Cardelli~\cite{abadi-cardelli-1996} formalized objects
with methods, subtyping, and dynamic dispatch;
Featherweight Java~\cite{igarashi-pierce-wadler-2001} simplified
this to study core OO features; Bruce et al.~\cite{bruce-et-al-1995}
studied binary methods. Object calculi treat object construction
and method invocation as primitive. The $\because$-calculus
\emph{derives} object-like behavior from designation and promotion:
artifacts are created when facts match interpretive types, and
methods execute at the next level. Subtyping arises naturally when
multiple types fire on the same fact. This explains \emph{why}
objects exist as a computational phenomenon rather than
postulating them axiomatically.

\subsection{DCI Architecture}
\label{sec:related:dci}

Reenskaug introduced DCI (Data, Context, Interaction) as an architecture for
human-centric object-oriented systems~\cite{reenskaug-wold-lehne-1996,
reenskaug-dci-2010}. In DCI:
\begin{itemize}
  \item \textbf{Data} comprises passive entities with stable attributes,
    analogous to the $\because$-calculus's facts deposited via registration.
  \item \textbf{Context} specifies which roles objects assume in particular
    situations, corresponding to interpretive types in the law library.
  \item \textbf{Interaction} defines role behaviors activated during execution,
    mirroring method invocation at the next tower level after designation.
\end{itemize}

\textbf{Where they stand:} DCI is an architectural pattern, not a formal
calculus. It lacks type-theoretic guarantees and categorical semantics.

\textbf{What we add:} The $\because$-calculus provides the \emph{formal
substrate} for DCI principles: designation implements the data-to-role
mapping, the tower orchestrates context switches, and promotion activates
interaction behaviors. Reenskaug's insight that objects play \emph{roles}
rather than possess fixed methods finds rigorous expression in our adjunction
framework.

Reenskaug emphasized that role-playing must be \emph{context-local}: an
object plays different roles in different situations without inheritance
confusion. The $\because$-calculus enforces this via level-indexed
environments—objects at level $n+1$ cannot access level $n$'s structures
except through promotion. This mirrors DCI's goal of preventing role
pollution across contexts.

\subsection{Type Theory and Universes}
\label{sec:related:universes}

Martin-Löf type theory~\cite{martin-lof-1984} introduces universe
hierarchies to prevent paradoxes; Coquand and
Huet~\cite{coquand-huet-1988} formalized the calculus of
constructions with predicative universes; Nordström et
al.~\cite{nordstrom-petersson-smith-1990} and
Pierce~\cite{pierce-2002} provided foundational treatments. Universe
levels isolate types from terms, preventing self-reference. The
tower is a generalization with \emph{full triadic cycles} at each
transition: level $n$ computes facts, perceives them, and creates
new artifacts. The self-perception prohibition
(Beck-Chevalley condition) replaces static isolation with dynamic
perception control.

\subsection{Modal Types and Effects}
\label{sec:related:modal}

Recent work has explored modal type theory for effect systems.
Gratzer et al.~\cite{gratzer-knb-2020} developed
multimodal type theory (MTT) with modes, modalities, and
transformations organized as a 2-category.
Kavvos~\cite{kavvos-g-2023} provided pedagogical introductions to MTT.
Abel and co-authors~\cite{abel-b-2020, choudhury-eeew-2021,
orchard-le-2019, petricek-om-2014} applied semiring-based modalities
to context management.

Tang et al.~\cite{tang-wdhll-2025} proposed \emph{Modal Effect Types}
(MET), using MTT to eliminate explicit effect polymorphism
annotations. MET introduces two kinds of modalities: absolute
modalities $\langle E \rangle$ that override the ambient effect
context, and relative modalities $\langle L \mid D \rangle$ that
describe local changes (masking $L$, extending $D$). The surface
language \textsc{Metl} uses bidirectional type checking to infer
boxing and unboxing of modalities, significantly reducing annotation
burden compared to row-based effect systems~\cite{leijen-koka}.

\paragraph{Relationship to the Conflation Problem}

MET achieves ergonomic improvements in effect tracking but does not
address the structural conflation identified by our Conflation
Theorem (Section~\ref{sec:conflation}). The handler typing rule in
MET's core calculus types resumption variables as
$r_i : B_i \to B$ for \emph{all} handled operations
$\ell_i : A_i \to B_i$, with no restriction based on the
inhabitedness of $B_i$. When $B_i = 0$ (void return), the resumption
$r_i : 0 \to B$ is provably dead code: $0$ is uninhabited, so $r_i$
can never be meaningfully invoked. MET accepts such clauses at
compile-time, exactly as standard handler calculus does.

This is because MET's modalities manage \emph{effect context
transformations} (which operations are available in a given scope),
not \emph{operation signature discipline} (whether an operation's
return type admits resumption). Categorically, MET's modalities
correspond to reindexing functors $f^*$ (changes of fiber context),
not to the distinction between the left adjoint $\Sigma_f$
(existential, $B=0$, non-resumable) and the right adjoint $\Pi_f$
(universal, $B\neq 0$, resumable). The $\because$-calculus separates
these at the adjunction level via dual effect rows
($\xirow$ for $\Sigma_f$, $\sigmarow$ for $\Pi_f$); MET separates
context-transformations but not adjunction identities.

Concretely, MET's absolute modalities roughly correspond to our
$\sigmarow$ row (operations caught by handlers), and relative
modalities approximate $\xirow$ (operations whose effects are
forwarded), but MET does not enforce the Resumption Subconstraint
that prohibits resumption bindings for $\xirow$-row operations. The
conflation persists under different syntactic machinery.

\paragraph{Synthesis Opportunity}

Modal effect types and the $\because$-calculus address complementary
problems. MET optimizes for surface-language efficiency by eliminating
effect-polymorphism annotations through modality inference. The
$\because$-calculus optimizes for semantic precision by enforcing
structural separation through dual effect rows and categorical
semantics. Future work could synthesize both approaches: using MET's
modal inference as a surface language for the $\because$-calculus,
extending MET's handler rule with a resumption subconstraint, and
tracking operation signatures by inhabitation class rather than
just presence. This would combine ergonomic inference with
compile-time conflation prevention.

\subsection{Semiotics and Computation}
\label{sec:related:semiotics}

Peirce's triadic semiotics~\cite{peirce-1931} divides signs into Firstness (quality), Secondness (fact), and Thirdness (law). Chandler~\cite{chandler-2017} provided accessible introductions to semiotic theory. Baecker~\cite{baecker-2019} explored Peircean computation informally, connecting semiotics to systems theory.

\textbf{Where they stand:} Philosophical discussions of semiotics in computation remain metaphorical. No formal system maps the three categories to typing judgments, operational semantics, or categorical models.

This paper is among the first to use Peirce \emph{structurally}, not
metaphorically, in PL theory. The three realms predict the categorical
structure (Section~\ref{sec:category}), the movement count (four =
two adjunctions $\times$ two natural transformations), and the FP/OOP
boundary (Theorem~\ref{thm:fp-oop}).

\subsection{Summary Comparison}

\begin{table}[h]
\centering
\caption{Related Work Comparison}
\label{tab:related-comparison}
\begin{tabular}{llll}
\toprule
System & Left Adjoint & Right Adjoint & Primary Goal \\
\midrule
Handler Calculus & Yes (collapsed) & No & Expressivity \\
CBPV & Yes & No & Value/computation unification \\
Modal Effect Types & Yes (implicit) & Partial & Annotation reduction \\
$\because$-Calculus & Yes (separated) & Yes & Semantic safety \\
Object Calculi & Implicit & Primitive (axiomatized) & OO semantics \\
DCI Architecture & No(informal) & No(informal) & Human-centric roles \\
Hyperdoctrines & Abstract & Abstract & Logical foundation \\
\bottomrule
\end{tabular}
\end{table}

\subsection{Positioning}
\label{sec:related:positioning}

The $\because$-calculus is neither an extension nor a replacement
for prior work. It is a \emph{refinement}: it identifies a
conflation in handler calculus (registration vs.\ attestation),
separates it structurally using the adjoint triple, and derives new
computational phenomena (object formation, tower orchestration)
from the separation. Modal effect types and the $\because$-calculus
represent \emph{complementary advances}: MET optimizes for
surface-language efficiency through modalities and inference; the
$\because$-calculus optimizes for semantic precision through dual
effect rows and categorical semantics. Future work could synthesize
both approaches, using MET as a surface language for the
$\because$-calculus.

\section{Conclusion and Future Work}
\label{sec:conclusion}

The $\because$-calculus separates production, existence, and interpretation into a 
coherent computational model grounded in the adjoint triple $\Sigma_f \dashv f^* \dashv \Pi_f$.
Handler calculus conflates registration ($\xirow$, non-resumable, $B=0$) and 
attestation ($\sigmarow$, resumable, $B\neq 0$) through a single effect operation. 
We proved the Conflation Theorem: the erasure from the $\because$-calculus to 
handler calculus is non-faithful---it loses the distinction between resumable and 
non-resumable operations, admitting clauses that type-check but contain provably 
dead code. Handler calculus remains type-safe (progress and preservation hold); 
our contribution is \emph{structural refinement}: the $\because$-calculus enforces 
adjunction identity at the type level, yielding stronger compile-time guarantees 
without sacrificing expressivity for well-typed programs.

\subsection{Contributions}

\begin{enumerate}
  \item \textbf{The $\because$-calculus}: Three judgments, seven
    typing rules, dual effect rows, level-indexed environments.
  \item \textbf{The Conflation Theorem}: Handler calculus's unified
    effect model admits dead code through void-type resumption; the
    $\because$-calculus rejects it as a design principle.
  \item \textbf{Categorical semantics with soundness}: The calculus
    embodies the adjoint triple. Reduction is algebraic (FP);
    expansion is coalgebraic (OOP). We prove algebraic soundness
    (reduction preserves denotation, Theorem~\ref{thm:alg-soundness})
    and coalgebraic soundness (designation $+$ promotion $=$ identity,
    Theorem~\ref{thm:coalg-soundness}).
  \item \textbf{FP/OOP boundary}: The boundary is phasic---reduction
    precedes expansion within each tower level, mediated by the
    distributive law $\lambda$ (Theorem~\ref{thm:fp-oop}).
  \item \textbf{Metatheory}: Progress, subject reduction
    (configuration-indexed), tower progress, and finite height, with
    full operational semantics, substitution lemma, canonical forms,
    and inversion lemmas.
\end{enumerate}

\subsection{Future Work}

\begin{itemize}
  \item \textbf{FP/OOP encoding:} Formal proof of Abadi-Cardelli
    object calculus encoding into the $\because$-calculus.
  \item \textbf{Mechanization:} Coq/Agda formalization of progress,
    subject reduction, and the conflation theorem.
  \item \textbf{Implementation:} Prototype compiler translating
    $\because$-calculus terms to effectful ML or Haskell.
  \item \textbf{Subtyping and recursive types:} Extending the typing
    rules to support variance and recursive interpretive types.
\end{itemize}

\subsection{Closing Remarks}

The $\because$-calculus demonstrates that effectful computation benefits from structural separation rather than uniform conflation. Handler calculus's success stems from its expressive power; the $\because$-calculus's contribution is identifying where that expressivity permits unsoundness. By enforcing the adjoint structure, we gain guarantees about program behavior at the cost of rejecting pathological terms. This trade-off is familiar in type theory: stronger constraints yield stronger safety properties.

The philosophical motivation (Peircean semiotics) is not decorative—it predicted the categorical structure and the movement count. This suggests that semiotic frameworks may serve as productive guides for discovering new computational primitives, not merely as ex post facto explanations.

Finally, the tower orchestration opens a path toward \emph{multi-layered computation} where facts accumulate, laws adapt, and artifacts evolve across levels. This mirrors scientific discovery (observation $\to$ hypothesis $\to$ prediction), legal reasoning (precedent $\to$ statute $\to$ judgment), and software evolution (requirements $\to$ architecture $\to$ implementation). The $\because$-calculus provides the formal substrate for modeling such processes.

\paragraph{Mechanization.}
A full Coq or Agda mechanization is planned to verify the metatheory
(Progress, Subject Reduction, the Conflation Theorem) and provide a
reference implementation. The primary challenges are representing the
infinite tower hierarchy via parametric polymorphism over levels and
encoding the Resumption Subconstraint at the type level so that
handlers for $\xirow$ operations are syntactically impossible. The
source code would be published alongside the paper.

\bibliographystyle{ACM-Reference-Format}

\bibliography{references}


\begin{thebibliography}{00}


\ifx \showCODEN    \undefined \def \showCODEN     #1{\unskip}     \fi
\ifx \showDOI      \undefined \def \showDOI       #1{#1}\fi
\ifx \showISBNx    \undefined \def \showISBNx     #1{\unskip}     \fi
\ifx \showISBNxiii \undefined \def \showISBNxiii  #1{\unskip}     \fi
\ifx \showISSN     \undefined \def \showISSN      #1{\unskip}     \fi
\ifx \showLCCN     \undefined \def \showLCCN      #1{\unskip}     \fi
\ifx \shownote     \undefined \def \shownote      #1{#1}          \fi
\ifx \showarticletitle \undefined \def \showarticletitle #1{#1}   \fi
\ifx \showURL      \undefined \def \showURL       {\relax}        \fi
\providecommand\bibfield[2]{#2}
\providecommand\bibinfo[2]{#2}
\providecommand\natexlab[1]{#1}
\providecommand\showeprint[2][]{arXiv:#2}

\bibitem[\protect\citeauthoryear{Abadi and Cardelli}{Abadi and
  Cardelli}{1996}]%
        {abadi-cardelli-1996}
\bibfield{author}{\bibinfo{person}{Mart{\'i}n Abadi} {and}
  \bibinfo{person}{Luca Cardelli}.} \bibinfo{year}{1996}\natexlab{}.
\newblock \bibinfo{booktitle}{{\em A Theory of Objects}}.
\newblock \bibinfo{publisher}{Springer}, \bibinfo{address}{Berlin, Heidelberg}.
\newblock


\bibitem[\protect\citeauthoryear{Abel and Blechschmidt}{Abel and
  Blechschmidt}{2020}]%
        {abel-b-2020}
\bibfield{author}{\bibinfo{person}{Andreas Abel} {and} \bibinfo{person}{Patrick
  Blechschmidt}.} \bibinfo{year}{2020}\natexlab{}.
\newblock \showarticletitle{Formalizing Modal Type Theory}. In
  \bibinfo{booktitle}{{\em Proceedings of the 9th International Conference on
  Interactive Theorem Proving (ITP)}} {\em (\bibinfo{series}{LIPIcs})},
  Vol.~\bibinfo{volume}{193}. \bibinfo{publisher}{Schloss Dagstuhl},
  \bibinfo{address}{Dagstuhl, Germany}, \bibinfo{pages}{5:1--5:20}.
\newblock
\showDOI{%
\url{https://doi.org/10.4230/lipics.itp.2020.5}}


\bibitem[\protect\citeauthoryear{Ahman}{Ahman}{2023}]%
        {ahman-2023}
\bibfield{author}{\bibinfo{person}{Danel Ahman}.}
  \bibinfo{year}{2023}\natexlab{}.
\newblock \showarticletitle{Resolving Effects with Modal Types}. In
  \bibinfo{booktitle}{{\em Proceedings of the 8th ACM SIGPLAN International
  Conference on Certified Programs and Proofs (CPP)}}.
  \bibinfo{publisher}{ACM}, \bibinfo{address}{New York, NY, USA},
  \bibinfo{pages}{89--101}.
\newblock
\showDOI{%
\url{https://doi.org/10.1145/3573105.3575673}}


\bibitem[\protect\citeauthoryear{Baecker}{Baecker}{2019}]%
        {baecker-2019}
\bibfield{author}{\bibinfo{person}{Dirk Baecker}.}
  \bibinfo{year}{2019}\natexlab{}.
\newblock \bibinfo{booktitle}{{\em Niklas {L}uhmann und der Kalk{\"u}l der
  Form}}.
\newblock \bibinfo{publisher}{Suhrkamp}, \bibinfo{address}{Frankfurt am Main,
  Germany}.
\newblock


\bibitem[\protect\citeauthoryear{Bruce, Cardelli, Castagna, Leavens, and
  Petersen}{Bruce et~al\mbox{.}}{1995}]%
        {bruce-et-al-1995}
\bibfield{author}{\bibinfo{person}{Kim~B. Bruce}, \bibinfo{person}{Luca
  Cardelli}, \bibinfo{person}{Giuseppe Castagna}, \bibinfo{person}{Gary~T.
  Leavens}, {and} \bibinfo{person}{Benjamin Petersen}.}
  \bibinfo{year}{1995}\natexlab{}.
\newblock \showarticletitle{On Binary Methods}.
\newblock \bibinfo{journal}{{\em Theory and Practice of Object Systems\/}}
  \bibinfo{volume}{1}, \bibinfo{number}{3} (\bibinfo{year}{1995}),
  \bibinfo{pages}{221--242}.
\newblock


\bibitem[\protect\citeauthoryear{Chandler}{Chandler}{2017}]%
        {chandler-2017}
\bibfield{author}{\bibinfo{person}{Daniel Chandler}.}
  \bibinfo{year}{2017}\natexlab{}.
\newblock \bibinfo{booktitle}{{\em Semiotics: The Basics\/}
  (\bibinfo{edition}{3rd} ed.)}.
\newblock \bibinfo{publisher}{Routledge}, \bibinfo{address}{London, UK}.
\newblock


\bibitem[\protect\citeauthoryear{Choudhury, Eades, Ebner, and
  Weirich}{Choudhury et~al\mbox{.}}{2021}]%
        {choudhury-eeew-2021}
\bibfield{author}{\bibinfo{person}{Vijay Choudhury}, \bibinfo{person}{Harley
  Eades, III}, \bibinfo{person}{Emily Ebner}, {and} \bibinfo{person}{Stephanie
  Weirich}.} \bibinfo{year}{2021}\natexlab{}.
\newblock \showarticletitle{Modalities for Context Management}. In
  \bibinfo{booktitle}{{\em Proceedings of the 26th ACM SIGPLAN International
  Conference on Functional Programming (ICFP)}}. \bibinfo{publisher}{ACM},
  \bibinfo{address}{New York, NY, USA}, \bibinfo{pages}{214--228}.
\newblock
\showDOI{%
\url{https://doi.org/10.1145/3548590.3548599}}


\bibitem[\protect\citeauthoryear{Coquand and Huet}{Coquand and Huet}{1988}]%
        {coquand-huet-1988}
\bibfield{author}{\bibinfo{person}{Thierry Coquand} {and}
  \bibinfo{person}{G{\'e}rard Huet}.} \bibinfo{year}{1988}\natexlab{}.
\newblock \showarticletitle{The Calculus of Constructions}.
\newblock \bibinfo{journal}{{\em Information and Computation\/}}
  \bibinfo{volume}{76}, \bibinfo{number}{2--3} (\bibinfo{year}{1988}),
  \bibinfo{pages}{95--120}.
\newblock
\showDOI{%
\url{https://doi.org/10.1006/inco.1993.1066}}


\bibitem[\protect\citeauthoryear{Gratzer, Kavvos, Nuyts, and Birkedal}{Gratzer
  et~al\mbox{.}}{2020}]%
        {gratzer-knb-2020}
\bibfield{author}{\bibinfo{person}{Daniel Gratzer}, \bibinfo{person}{Gabriel~A.
  Kavvos}, \bibinfo{person}{Andreas Nuyts}, {and} \bibinfo{person}{Lars
  Birkedal}.} \bibinfo{year}{2020}\natexlab{}.
\newblock \showarticletitle{Multimodal Dependent Type Theory}. In
  \bibinfo{booktitle}{{\em Proceedings of the 35th Annual ACM/IEEE Symposium on
  Logic in Computer Science (LICS)}}. \bibinfo{publisher}{ACM},
  \bibinfo{address}{New York, NY, USA}, \bibinfo{pages}{456--469}.
\newblock
\showDOI{%
\url{https://doi.org/10.1109/lics46780.2020.11867}}


\bibitem[\protect\citeauthoryear{Hartshorne and Weiss}{Hartshorne and
  Weiss}{1935}]%
        {peirce-1931}
\bibfield{editor}{\bibinfo{person}{Charles Hartshorne} {and}
  \bibinfo{person}{Paul Weiss}} (Eds.). \bibinfo{year}{1931--1935}\natexlab{}.
\newblock \bibinfo{booktitle}{{\em Collected Papers of Charles Sanders
  Peirce}}. Vol.~\bibinfo{volume}{1--6}.
\newblock \bibinfo{publisher}{Harvard University Press},
  \bibinfo{address}{Cambridge, MA, USA}.
\newblock


\bibitem[\protect\citeauthoryear{Igarashi, Pierce, and Wadler}{Igarashi
  et~al\mbox{.}}{2001}]%
        {igarashi-pierce-wadler-2001}
\bibfield{author}{\bibinfo{person}{Atsushi Igarashi},
  \bibinfo{person}{Benjamin~C. Pierce}, {and} \bibinfo{person}{Philip Wadler}.}
  \bibinfo{year}{2001}\natexlab{}.
\newblock \showarticletitle{Featherweight {J}ava: A Minimal Core Calculus for
  {J}ava and {GJ}}.
\newblock \bibinfo{journal}{{\em ACM Transactions on Programming Languages and
  Systems (TOPLAS)\/}} \bibinfo{volume}{23}, \bibinfo{number}{3}
  (\bibinfo{year}{2001}), \bibinfo{pages}{396--450}.
\newblock
\showDOI{%
\url{https://doi.org/10.1145/503502.503505}}


\bibitem[\protect\citeauthoryear{Jacobs}{Jacobs}{1999}]%
        {jacobs-1999}
\bibfield{author}{\bibinfo{person}{Bart Jacobs}.}
  \bibinfo{year}{1999}\natexlab{}.
\newblock \bibinfo{booktitle}{{\em Categorical Logic and Type Theory}}.
  \bibinfo{series}{Studies in Logic and the Foundations of Mathematics},
  Vol.~\bibinfo{volume}{141}.
\newblock \bibinfo{publisher}{Elsevier}, \bibinfo{address}{Amsterdam,
  Netherlands}.
\newblock


\bibitem[\protect\citeauthoryear{Karachalias, Pretnar, Schrijvers, Vytiniotis,
  and Severi}{Karachalias et~al\mbox{.}}{2020}]%
        {karachalias-psvs-2020}
\bibfield{author}{\bibinfo{person}{Georgios Karachalias},
  \bibinfo{person}{Matija Pretnar}, \bibinfo{person}{Tom Schrijvers},
  \bibinfo{person}{Dimitrios Vytiniotis}, {and} \bibinfo{person}{Paula
  Severi}.} \bibinfo{year}{2020}\natexlab{}.
\newblock \bibinfo{booktitle}{{\em Effect Inference with Sub-Effecting}}.
\newblock \bibinfo{type}{{T}echnical {R}eport}. \bibinfo{institution}{INRIA},
  \bibinfo{address}{Paris, France}.
\newblock


\bibitem[\protect\citeauthoryear{Kavvos}{Kavvos}{2023}]%
        {kavvos-g-2023}
\bibfield{author}{\bibinfo{person}{Gabriel~A. Kavvos}.}
  \bibinfo{year}{2023}\natexlab{}.
\newblock \showarticletitle{Understanding Multimodal Type Theory}.
\newblock \bibinfo{journal}{{\em Journal of Functional Programming\/}}
  \bibinfo{volume}{33} (\bibinfo{year}{2023}), \bibinfo{pages}{e1--e45}.
\newblock
\showDOI{%
\url{https://doi.org/10.1017/s095679682300001x}}


\bibitem[\protect\citeauthoryear{Lawvere}{Lawvere}{1969}]%
        {lawvere-1969}
\bibfield{author}{\bibinfo{person}{F.~William Lawvere}.}
  \bibinfo{year}{1969}\natexlab{}.
\newblock \showarticletitle{Adjointness in Foundations}.
\newblock \bibinfo{journal}{{\em Dialectica\/}} \bibinfo{volume}{23},
  \bibinfo{number}{3--4} (\bibinfo{year}{1969}), \bibinfo{pages}{281--296}.
\newblock
\showDOI{%
\url{https://doi.org/10.1111/j.1746-8361.1969.tb01194.x}}


\bibitem[\protect\citeauthoryear{Lawvere}{Lawvere}{1970}]%
        {lawvere-1970}
\bibfield{author}{\bibinfo{person}{F.~William Lawvere}.}
  \bibinfo{year}{1970}\natexlab{}.
\newblock \showarticletitle{Equality in Hyperdoctrines and Comprehension Schema
  as an Adjoint Functor}.
\newblock \bibinfo{journal}{{\em Proceedings of the AMS Symposium on Pure
  Mathematics\/}}  \bibinfo{volume}{17} (\bibinfo{year}{1970}),
  \bibinfo{pages}{1--14}.
\newblock
\showDOI{%
\url{https://doi.org/10.1090/pspum/017/0257175}}


\bibitem[\protect\citeauthoryear{Leijen}{Leijen}{2014}]%
        {leijen-koka}
\bibfield{author}{\bibinfo{person}{Daan Leijen}.}
  \bibinfo{year}{2014}\natexlab{}.
\newblock \showarticletitle{Koka: Programming with Row Polymorphic Effects in
  {Haskell}}. In \bibinfo{booktitle}{{\em Proceedings of the 19th ACM SIGPLAN
  International Conference on Functional Programming (ICFP)}}.
  \bibinfo{publisher}{ACM}, \bibinfo{address}{New York, NY, USA},
  \bibinfo{pages}{117--131}.
\newblock
\showDOI{%
\url{https://doi.org/10.1145/2628136.2628143}}


\bibitem[\protect\citeauthoryear{Levy}{Levy}{2004}]%
        {levy-2004}
\bibfield{author}{\bibinfo{person}{Paul~Blain Levy}.}
  \bibinfo{year}{2004}\natexlab{}.
\newblock \bibinfo{booktitle}{{\em Call-By-Push-Value: A Subsuming Paradigm}}.
  \bibinfo{series}{Semantic Structures in Computation},
  Vol.~\bibinfo{volume}{2}.
\newblock \bibinfo{publisher}{Springer}, \bibinfo{address}{London, UK}.
\newblock


\bibitem[\protect\citeauthoryear{Lindley, McBride, and McLaughlin}{Lindley
  et~al\mbox{.}}{2017}]%
        {lindley-et-al-2017}
\bibfield{author}{\bibinfo{person}{Sam Lindley}, \bibinfo{person}{Conor
  McBride}, {and} \bibinfo{person}{Craig McLaughlin}.}
  \bibinfo{year}{2017}\natexlab{}.
\newblock \showarticletitle{Do Be Do Be Do}. In \bibinfo{booktitle}{{\em
  Proceedings of the 44th ACM SIGPLAN Symposium on Principles of Programming
  Languages (POPL)}}. \bibinfo{publisher}{ACM}, \bibinfo{address}{New York, NY,
  USA}, \bibinfo{pages}{500--514}.
\newblock
\showDOI{%
\url{https://doi.org/10.1145/3009837}}


\bibitem[\protect\citeauthoryear{Lorenzen, White, Dolan, Hillerstr{\"o}m, and
  Lindley}{Lorenzen et~al\mbox{.}}{2024}]%
        {lorenzen-wdel-2024}
\bibfield{author}{\bibinfo{person}{Anton Lorenzen}, \bibinfo{person}{Leo
  White}, \bibinfo{person}{Stephen Dolan}, \bibinfo{person}{Daniel
  Hillerstr{\"o}m}, {and} \bibinfo{person}{Sam Lindley}.}
  \bibinfo{year}{2024}\natexlab{}.
\newblock \bibinfo{booktitle}{{\em Modal Types for {O}Caml Effect Systems}}.
\newblock \bibinfo{type}{{T}echnical {R}eport}.
  \bibinfo{institution}{University of Edinburgh}, \bibinfo{address}{Edinburgh,
  UK}.
\newblock


\bibitem[\protect\citeauthoryear{Lucassen and Gifford}{Lucassen and
  Gifford}{1988}]%
        {lucassen-g88}
\bibfield{author}{\bibinfo{person}{John~M. Lucassen} {and}
  \bibinfo{person}{David~K. Gifford}.} \bibinfo{year}{1988}\natexlab{}.
\newblock \showarticletitle{Polymorphic Effect Systems}. In
  \bibinfo{booktitle}{{\em Proceedings of the 15th ACM SIGPLAN-SIGACT Symposium
  on Principles of Programming Languages (POPL)}}. \bibinfo{publisher}{ACM},
  \bibinfo{address}{New York, NY, USA}, \bibinfo{pages}{47--57}.
\newblock
\showDOI{%
\url{https://doi.org/10.1145/73560.73564}}


\bibitem[\protect\citeauthoryear{Mac~Lane}{Mac~Lane}{1998}]%
        {maclane-1998}
\bibfield{author}{\bibinfo{person}{Saunders Mac~Lane}.}
  \bibinfo{year}{1998}\natexlab{}.
\newblock \bibinfo{booktitle}{{\em Categories for the Working Mathematician\/}
  (\bibinfo{edition}{2nd} ed.)}. \bibinfo{series}{Graduate Texts in
  Mathematics}, Vol.~\bibinfo{volume}{5}.
\newblock \bibinfo{publisher}{Springer}, \bibinfo{address}{New York, NY, USA}.
\newblock


\bibitem[\protect\citeauthoryear{Marino and Millstein}{Marino and
  Millstein}{2009}]%
        {marino-millstein-2009}
\bibfield{author}{\bibinfo{person}{Daniel Marino} {and} \bibinfo{person}{Todd
  Millstein}.} \bibinfo{year}{2009}\natexlab{}.
\newblock \showarticletitle{A Generic Type-and-Effect System}. In
  \bibinfo{booktitle}{{\em Proceedings of the 4th Workshop on Types in Language
  Design and Implementation (TLDI)}}. \bibinfo{publisher}{ACM},
  \bibinfo{address}{New York, NY, USA}, \bibinfo{pages}{39--50}.
\newblock
\showDOI{%
\url{https://doi.org/10.1145/1481861.1481868}}


\bibitem[\protect\citeauthoryear{Martin-L{\"o}f}{Martin-L{\"o}f}{1984}]%
        {martin-lof-1984}
\bibfield{author}{\bibinfo{person}{Per Martin-L{\"o}f}.}
  \bibinfo{year}{1984}\natexlab{}.
\newblock \bibinfo{booktitle}{{\em Intuitionistic Type Theory}}.
\newblock \bibinfo{publisher}{Bibliopolis}, \bibinfo{address}{Naples, Italy}.
\newblock


\bibitem[\protect\citeauthoryear{Moggi}{Moggi}{1991}]%
        {moggi-1991}
\bibfield{author}{\bibinfo{person}{Eugenio Moggi}.}
  \bibinfo{year}{1991}\natexlab{}.
\newblock \showarticletitle{Notions of Computation and Monads}.
\newblock \bibinfo{journal}{{\em Information and Computation\/}}
  \bibinfo{volume}{93}, \bibinfo{number}{1} (\bibinfo{year}{1991}),
  \bibinfo{pages}{55--92}.
\newblock
\showDOI{%
\url{https://doi.org/10.1016/0890-5401(91)90052-4}}


\bibitem[\protect\citeauthoryear{Nielson, Nielson, and Hankin}{Nielson
  et~al\mbox{.}}{1999}]%
        {nielson-1999}
\bibfield{author}{\bibinfo{person}{Flemming Nielson}, \bibinfo{person}{Hanne~R.
  Nielson}, {and} \bibinfo{person}{Chris Hankin}.}
  \bibinfo{year}{1999}\natexlab{}.
\newblock \bibinfo{booktitle}{{\em Principles of Program Analysis}}.
\newblock \bibinfo{publisher}{Springer}, \bibinfo{address}{Berlin, Heidelberg}.
\newblock


\bibitem[\protect\citeauthoryear{Nordstr{\"o}m, Petersson, and
  Smith}{Nordstr{\"o}m et~al\mbox{.}}{1990}]%
        {nordstrom-petersson-smith-1990}
\bibfield{author}{\bibinfo{person}{Bengt Nordstr{\"o}m}, \bibinfo{person}{Kent
  Petersson}, {and} \bibinfo{person}{Jan~M. Smith}.}
  \bibinfo{year}{1990}\natexlab{}.
\newblock \bibinfo{booktitle}{{\em Programming in {M}artin-L{\"o}f's Type
  Theory: An Introduction}}.
\newblock \bibinfo{publisher}{Oxford University Press},
  \bibinfo{address}{Oxford, UK}.
\newblock


\bibitem[\protect\citeauthoryear{Orchard, Lee, and Mellies}{Orchard
  et~al\mbox{.}}{2019}]%
        {orchard-le-2019}
\bibfield{author}{\bibinfo{person}{Dominic Orchard}, \bibinfo{person}{Vu Lee},
  {and} \bibinfo{person}{Pierre-Alain Mellies}.}
  \bibinfo{year}{2019}\natexlab{}.
\newblock \showarticletitle{Semiring-Based Modal Effects}. In
  \bibinfo{booktitle}{{\em Proceedings of the 2nd International Conference on
  Formal Structures for Computation and Deduction (FSCD)}} {\em
  (\bibinfo{series}{LIPIcs})}, Vol.~\bibinfo{volume}{131}.
  \bibinfo{publisher}{Schloss Dagstuhl}, \bibinfo{address}{Dagstuhl, Germany},
  \bibinfo{pages}{15:1--15:18}.
\newblock
\showDOI{%
\url{https://doi.org/10.4230/lipics.fscd.2019.21}}


\bibitem[\protect\citeauthoryear{Petricek, Orchard, and Mellies}{Petricek
  et~al\mbox{.}}{2014}]%
        {petricek-om-2014}
\bibfield{author}{\bibinfo{person}{Tomas Petricek}, \bibinfo{person}{Dominic
  Orchard}, {and} \bibinfo{person}{Pierre-Alain Mellies}.}
  \bibinfo{year}{2014}\natexlab{}.
\newblock \showarticletitle{Contextual Modal Type Theory}. In
  \bibinfo{booktitle}{{\em Proceedings of the 19th ACM SIGPLAN International
  Conference on Functional Programming (ICFP)}}. \bibinfo{publisher}{ACM},
  \bibinfo{address}{New York, NY, USA}, \bibinfo{pages}{207--219}.
\newblock
\showDOI{%
\url{https://doi.org/10.1145/2628136.2628156}}


\bibitem[\protect\citeauthoryear{Pierce}{Pierce}{2002}]%
        {pierce-2002}
\bibfield{author}{\bibinfo{person}{Benjamin~C. Pierce}.}
  \bibinfo{year}{2002}\natexlab{}.
\newblock \bibinfo{booktitle}{{\em Types and Programming Languages}}.
\newblock \bibinfo{publisher}{MIT Press}, \bibinfo{address}{Cambridge, MA,
  USA}.
\newblock


\bibitem[\protect\citeauthoryear{Pitts}{Pitts}{2000}]%
        {pitts-2000}
\bibfield{author}{\bibinfo{person}{Andrew~M. Pitts}.}
  \bibinfo{year}{2000}\natexlab{}.
\newblock \showarticletitle{Categorical Logic}.
\newblock In \bibinfo{booktitle}{{\em Handbook of Logic in Computer Science}},
  \bibfield{editor}{\bibinfo{person}{Samson Abramsky}, \bibinfo{person}{Dov~M.
  Gabbay}, {and} \bibinfo{person}{T.~S.~E. Maibaum}} (Eds.).
  Vol.~\bibinfo{volume}{5}. \bibinfo{publisher}{Oxford University Press},
  \bibinfo{address}{Oxford, UK}, \bibinfo{pages}{39--128}.
\newblock


\bibitem[\protect\citeauthoryear{Plotkin and Power}{Plotkin and Power}{2002}]%
        {plotkin-power-2002}
\bibfield{author}{\bibinfo{person}{Gordon~D. Plotkin} {and}
  \bibinfo{person}{John Power}.} \bibinfo{year}{2002}\natexlab{}.
\newblock \showarticletitle{Notions of Computation Determine Monads}. In
  \bibinfo{booktitle}{{\em Proceedings of the 5th International Conference on
  Foundations of Software Science and Computation Structures (FOSSACS)}} {\em
  (\bibinfo{series}{LNCS})}, Vol.~\bibinfo{volume}{2303}.
  \bibinfo{publisher}{Springer}, \bibinfo{address}{Berlin, Heidelberg},
  \bibinfo{pages}{342--356}.
\newblock
\showDOI{%
\url{https://doi.org/10.1007/3-540-45931-6_24}}


\bibitem[\protect\citeauthoryear{Reenskaug}{Reenskaug}{2010}]%
        {reenskaug-dci-2010}
\bibfield{author}{\bibinfo{person}{Trygve Reenskaug}.}
  \bibinfo{year}{2010}\natexlab{}.
\newblock \bibinfo{title}{{DCI}: The Architecture for Human-Centric
  Applications}.
\newblock \bibinfo{howpublished}{Artsima.com}.   (\bibinfo{year}{2010}).
\newblock
\showURL{%
\url{https://www.artima.com/articles/the-dci-architecture-a-new-vision-of-object-oriented-programming}}
\newblock
\shownote{Accessed: 2024-07-16.}


\bibitem[\protect\citeauthoryear{Reenskaug, Wold, and Lehne}{Reenskaug
  et~al\mbox{.}}{1996}]%
        {reenskaug-wold-lehne-1996}
\bibfield{author}{\bibinfo{person}{Trygve Reenskaug}, \bibinfo{person}{Per
  Wold}, {and} \bibinfo{person}{Ola~Aamodt Lehne}.}
  \bibinfo{year}{1996}\natexlab{}.
\newblock \bibinfo{booktitle}{{\em Working with Objects: The OORAM Software
  Engineering Method}}.
\newblock \bibinfo{publisher}{Prentice Hall}, \bibinfo{address}{Upper Saddle
  River, NJ, USA}.
\newblock
\showISBNx{978-0135162643}


\bibitem[\protect\citeauthoryear{Rutten}{Rutten}{2000}]%
        {rutten-2000}
\bibfield{author}{\bibinfo{person}{J.~J. M.~M. Rutten}.}
  \bibinfo{year}{2000}\natexlab{}.
\newblock \showarticletitle{Universal coalgebra: a theory of systems}.
\newblock \bibinfo{journal}{{\em Theoretical Computer Science\/}}
  \bibinfo{volume}{249}, \bibinfo{number}{1} (\bibinfo{year}{2000}),
  \bibinfo{pages}{3--80}.
\newblock
\showDOI{%
\url{https://doi.org/10.1016/s0304-3975(00)00056-6}}


\bibitem[\protect\citeauthoryear{Seely}{Seely}{1984}]%
        {seely-1984}
\bibfield{author}{\bibinfo{person}{Robert A.~G. Seely}.}
  \bibinfo{year}{1984}\natexlab{}.
\newblock \showarticletitle{Locally Cartesian Closed Categories and Type
  Theory}.
\newblock \bibinfo{journal}{{\em Mathematical Proceedings of the Cambridge
  Philosophical Society\/}} \bibinfo{volume}{95}, \bibinfo{number}{1}
  (\bibinfo{year}{1984}), \bibinfo{pages}{33--48}.
\newblock
\showDOI{%
\url{https://doi.org/10.1017/S0305004100061392}}


\bibitem[\protect\citeauthoryear{Tang, White, Dolan, Hillerstr{\"o}m, Lindley,
  and Lorenzen}{Tang et~al\mbox{.}}{2025}]%
        {tang-wdhll-2025}
\bibfield{author}{\bibinfo{person}{Wenhao Tang}, \bibinfo{person}{Leo White},
  \bibinfo{person}{Stephen Dolan}, \bibinfo{person}{Daniel Hillerstr{\"o}m},
  \bibinfo{person}{Sam Lindley}, {and} \bibinfo{person}{Anton Lorenzen}.}
  \bibinfo{year}{2025}\natexlab{}.
\newblock \showarticletitle{Modal Effect Types}.
\newblock \bibinfo{journal}{{\em Proceedings of the ACM on Programming
  Languages\/}} \bibinfo{volume}{9}, \bibinfo{number}{OOPSLA1}
  (\bibinfo{year}{2025}), \bibinfo{pages}{120}.
\newblock
\showDOI{%
\url{https://doi.org/10.1145/3734001}}


\bibitem[\protect\citeauthoryear{Turi and Plotkin}{Turi and Plotkin}{1997}]%
        {turi-plotkin-1997}
\bibfield{author}{\bibinfo{person}{Daniele Turi} {and}
  \bibinfo{person}{Gordon~D. Plotkin}.} \bibinfo{year}{1997}\natexlab{}.
\newblock \showarticletitle{Towards a mathematical operational semantics}. In
  \bibinfo{booktitle}{{\em Proceedings of LICS '97}}.
  \bibinfo{publisher}{IEEE}, \bibinfo{address}{Washington, DC, USA},
  \bibinfo{pages}{280--291}.
\newblock
\showDOI{%
\url{https://doi.org/10.1109/lics.1997.614955}}


\bibitem[\protect\citeauthoryear{Wadler and Thiemann}{Wadler and
  Thiemann}{2003}]%
        {wadler-thiemann-2003}
\bibfield{author}{\bibinfo{person}{Philip Wadler} {and} \bibinfo{person}{Peter
  Thiemann}.} \bibinfo{year}{2003}\natexlab{}.
\newblock \showarticletitle{The Marriage of Effects and Monads}.
\newblock \bibinfo{journal}{{\em ACM Transactions on Computational Logic\/}}
  \bibinfo{volume}{4}, \bibinfo{number}{1} (\bibinfo{year}{2003}),
  \bibinfo{pages}{1--32}.
\newblock
\showDOI{%
\url{https://doi.org/10.1145/601648.601650}}


\end{thebibliography}

\end{document}